%% file: TNNLS.tex
\documentclass[journal]{IEEEtran}
\RequirePackage{pdf14} 
\usepackage{cite}
\usepackage{graphicx}
\usepackage{subcaption}
\usepackage[cmex10]{amsmath}
\usepackage{xcolor}
\usepackage{array}
\usepackage{url}
\usepackage{bbm}
\usepackage{todonotes}

\usepackage{algorithm}
\usepackage{algorithmic}
\usepackage{bm}
\renewcommand{\algorithmicrequire}{\textbf{Input:}}

\usepackage{amsmath,amssymb}

\usepackage{amsmath} 
\usepackage{bm}      
\usepackage{epsfig}
\usepackage{amsfonts,amssymb,multirow,bigstrut,booktabs,ctable,latexsym}
\usepackage{mathtools}
\usepackage[mathscr]{eucal}
\usepackage{verbatim}

\usepackage{amsthm}
\usepackage{algorithm}
\usepackage[normalem]{ulem}

\newcommand{\argmax}{\arg\!\max}

\newcommand{\norm}[1]{\left\lVert#1\right\rVert}

\usepackage{amsmath}
\newtheorem{thm}{Theorem}[section]

\newtheorem{lm}[thm]{Lemma}
\newtheorem{cor}[thm]{Corollary}
\newtheorem{df}[thm]{Definition}

\newtheorem{rem}[thm]{Remark}
\newtheorem{prop}[thm]{Proposition}

\newtheorem{ass}[thm]{Assumption}

\usepackage{amssymb}
\usepackage[switch]{lineno}
\usepackage{amsthm}

\IEEEoverridecommandlockouts
\begin{document}
	\title{\bf Shaping Advice in Deep Reinforcement Learning}
	\author{Baicen Xiao$^1$, \emph{Member, IEEE}, Bhaskar Ramasubramanian$^{1,2}$, \emph{Member, IEEE}, Radha Poovendran$^1$, \emph{Fellow, IEEE}
{		\thanks{$^{1}$Network Security Lab, Department of Electrical and Computer Engineering, 
			University of Washington, Seattle, WA 98195, USA. \newline
			{\tt\small \{bcxiao, rp3\}@uw.edu}}
		\thanks{$^{2}$Electrical and Computer Engineering, Department of Engineering and Design, 
			Western Washington University, Bellingham, WA 98225, USA. \newline
			{\tt\small \{ramasub\}@wwu.edu}}
		\thanks{This work was supported by the Office of Naval Research via Grant N00014-17-S-B001.}
	}}
\maketitle
	
\begin{abstract}
 Reinforcement learning involves agents interacting with an environment to complete tasks. 
When rewards provided by the environment are sparse, agents may not receive immediate feedback on the quality of actions that they take, thereby affecting learning of policies. 
In this paper, we propose methods to augment the reward signal from the environment with an additional reward termed \emph{shaping advice} in both single- and multi-agent reinforcement learning. 
The shaping advice is specified as a difference of potential functions at consecutive time-steps. 
Each potential function is a function of observations and actions of the agents. 
The use of potential functions is underpinned by an insight that the total potential when starting from any state and returning to the same state is always equal to zero.  
We show through theoretical analyses and experimental validation that the shaping advice does not distract agents from completing tasks specified by the environment reward. 
Theoretically, we prove that the convergence of policy gradients and value functions when using shaping advice implies the convergence of these quantities 
in the absence of shaping advice. 
We design two algorithms- \emph{Shaping Advice in Single-agent reinforcement learning (SAS)} and \emph{Shaping Advice in Multi-agent reinforcement learning (SAM)}. 
Shaping advice in \emph{SAS} and \emph{SAM} needs to be specified only once at the start of training, and can easily be provided by non-experts. 
Experimentally, we evaluate \emph{SAS} and \emph{SAM} on two tasks in single-agent environments and three tasks in multi-agent environments that have sparse rewards. We observe that using shaping advice results in agents learning policies to complete tasks faster, and obtain higher rewards than algorithms that do not use shaping advice. 
Code for our experiments is available at \textcolor{blue}{\url{https://github.com/baicenxiao/Shaping-Advice}}. 
\end{abstract}
	
\input{Introduction}
\input{Related_work}
\input{Background}

\input{Methods}
\input{Experiment}
\input{Conclusion}

\bibliographystyle{IEEEtran}
\bibliography{RewShapBib}

\end{document}

%% file: Introduction.tex
\section{Introduction}
 
Reinforcement learning (RL) is a framework that allows agents to complete tasks in an environment, even when a model of the environment is not known \cite{sutton2018reinforcement}. 
An RL agent `learns' to complete tasks by maximizing an expected long-term reward, where the reward signal is provided by the environment. 
RL algorithms have been successfully implemented in many fields, including games \cite{mnih2015human, silver2016mastering, tampuu2017multiagent}, robotics \cite{lillicrap2016continuous}, autonomous vehicle coordination \cite{sallab2017deep}, analysis of social dilemmas \cite{leibo2017multi}, and resource allocation in cellular networks \cite{yin2021resource}. 

When the environment is not known, availability of immediate feedback on quality of actions taken 
at each time-step is critical to the learning of behaviors to successfully complete a task. 
This is termed \emph{credit assignment} \cite{sutton2018reinforcement}. 
When reward signals provided by the environment are sparse, it becomes difficult to perform effective credit assignment at intermediate steps of the learning process. 
One approach that has been shown to improve learning when rewards are sparse is \emph{reward shaping} \cite{agogino2008analyzing, devlin2011empirical, devlin2014potential}. 
Reward shaping techniques augment the reward provided by the environment with an additional \emph{shaping reward}. 
The shaping reward can be designed to be \emph{dense} (i.e., not sparse), and agents learn policies (which action to take in a particular state) using the augmented reward.

Any additional reward can distract an agent from completing a task specified by the reward provided by the environment, and therefore will need to be provided in a systematic manner \cite{randlov1998learning}. 
In this paper, we term the additional reward given to agents at each time-step as \emph{\textbf{shaping advice}}. 
The shaping advice is specified by a difference of potential functions at consecutive time-steps, where each potential function depends on observations and actions of the agents. 
Potential functions satisfy a critical property that the total potential when starting from a state and returning to the same state is zero \cite{ng1999policy}. 
This ensures that an agent will not be distracted from completing a task specified by the reward provided by the environment. 

There are additional challenges when there are multiple RL agents. 
In such a setting, each agent will have to interact not only with its environment, but also with other agents.
As behaviors of agents evolve in the environment, the environment will become non-stationary from the perspective of any single agent. 
Thus, agents that independently learn behaviors by assuming other agents to be part of the environment can result in unstable learning regimes \cite{foerster2017stabilising, matignon2012independent, tan1993multi}. 

When multiple trained agents are deployed independently, or when communication among agents is costly, the agents will need to be able to learn decentralized policies. 
Decentralized policies can be efficicently learned by adopting the \emph{centralized training with decentralized execution (CTDE)} paradigm, first introduced in \cite{lowe2017multi}. 
An agent using CTDE can make use of information about other agents' observations and actions to aid its own learning during training, but will have to take decisions independently at test-time. 
However, the ability of an agent to learn decentralized policies can be affected if reward signals from the environment are sparse. 


Reward shaping techniques that use potential functions satisfy an additional property that the identity of the optimal policies with and without the shaping reward is the same \cite{ng1999policy}. 
The state-of-the-art on the potential-based reward shaping in RL \cite{agogino2008analyzing, devlin2011empirical, devlin2014potential} is focused on environments with finite and discrete action spaces. 
To the best of our knowledge, adapting solution techniques proposed in the above papers to more general settings with continuous action spaces has not been explored. 
Moreover, in the multi-agent case, these works emphasize learning joint policies for agents. 
In comparison, we propose to learn decentralized policies for each agent, which will enable application of our method to environments with large numbers of agents. 

In this paper, we develop a comprehensive framework to enable effective credit assignment in RL environments with sparse rewards. 
These methods incorporate information about the task and environment to define \emph{shaping advice}. 
We term our algorithms \emph{Shaping Advice in Single-agent reinforcement learning (SAS)} when there is only a single agent, and \emph{Shaping Advice in Multi-agent reinforcement learning (SAM)}, when there are multiple agents. 
The shaping advice in SAS and SAM can be interpreted as \emph{domain knowledge} that aids credit assignment \cite{mannion2018reward}. 
This advice needs to be specified only once at the start of the training process. 
SAS and SAM can be applied in environments with continuous or discrete state and action spaces. 
We demonstrate that both algorithms do not distract agents from completing tasks specified by rewards from the environment. 
Specifically, our contributions are:
\begin{itemize}
\item We introduce SAS and SAM to incorporate potential-based shaping advice in single- and multi-agent deep RL environments with continuous action spaces. 
\item We demonstrate that shaping advice provided by SAS and SAM does not distract agents from completing tasks specified by the environment reward. We accomplish this by theoretically establishing that  
convergence of policy gradients and values when using the shaping advice implies convergence of these quantities in the absence of the shaping advice. 
\item We verify our theoretical results through extensive experimental evaluations. Specifically, 
\begin{itemize}
\item we evaluate SAS on two environments with sparse rewards- puddle-jump gridworld and continuous mountain car; 
\item we evaluate SAM on three tasks in the multi-agent Particle World environment \cite{lowe2017multi}. All these tasks have sparse rewards. We show that using shaping advice allows agents to learn policies to complete the tasks faster, and obtain higher rewards than: i) using sparse rewards alone, and ii) a state-of-the-art reward redistribution technique from \cite{gangwani2020learning}. 
\end{itemize}
\end{itemize}
Compared to a preliminary version that appeared in \cite{xiao2019potential}, in this paper, we develop a comprehensive framework for providing shaping advice in both, single- and multi-agent RL. 
We provide detailed theoretical analyses and experimental evaluations for each setting. 


The remainder of this paper is organized as follows: Section \ref{RelatedWork} presents related work and Section \ref{Background} provides an introduction to single- and multi-agent RL and potential-based shaping advice. 
Sections \ref{Methods1} - \ref{Experiments} present the main contributions of this paper. 
We provide details on SAS and SAM and present theoretical analysis of their convergence in Sections \ref{Methods1} and \ref{Methods2}. Experiments validating the use of SAS and SAM are reported in Section \ref{Experiments}, and Section \ref{Conclusion} concludes the paper.

%% file: Related_work.tex
\section{Related Work}\label{RelatedWork}

Techniques to improve credit assignment using feedback signals provided by a human operator have been studied in single-agent RL. 
Demonstrations provided by a human operator were used to synthesize a ‘baseline policy’ that was used to guide learning in \cite{taylor2011integrating, wang2017improving}. 
When expert demonstrations were available, imitation learning was used to guide exploration of the RL agent in \cite{kelly2019hg, ross2011reduction}. 
Feedback provided by a human operator was converted to a shaping reward to aid training a deep RL agent in environments with delayed rewards in \cite{xiao2020fresh}. 
These techniques, however, presume the availability of a human operator, which might limit their applicability. 

In the multi-agent setting, decentralized and distributed control techniques is a popular area of research. 
A widely studied problem in such systems is the development of algorithms to specify methods by which information can be exchanged among agents so that they can jointly complete tasks. 
A technique to ensure fixed-time consensus for multi-agent systems whose interactions were specified by a directed graph was studied in \cite{zuo2015nonsingular, tian2018fixed}. 
The authors of \cite{li2020consensus} proposed an adaptive distributed event-triggering protocol to guarantee consensus for multi-agent systems specified by linear dynamics and interactions specified by an undirected graph. 
We direct the reader to \cite{qin2016recent} for a survey of recent developments in consensus of multi-agent systems. 
These works, however, assumed the availability of models of individual agent's dynamics, and of the interactions between agents. 

Data-driven techniques are being increasingly adopted to solve problems in multi-agent systems. 
These methods do not require complete knowledge of the system model; rather, they use information about the state, input, and output of each agent to establish feasibility of solutions and guarantees on convergence of reinforcement learning-based algorithms. 
Input and state data were used in an online manner to design a distributed control algorithm to solve a cooperative optimal output regulation problem in leader-follower systems in \cite{gao2021reinforcement}. 
Information obtained from trajectories of each player were used in \cite{vamvoudakis2017game} to develop real-time solutions to multi-player games through the design of an actor-critic-based adaptive learning algorithm.  
The authors of \cite{qu2020scalable} identified a class of networked MARL problems where enforcing specific local interactions among players permitted the exploiting an \emph{exponential decay property} which enabled the development of scalable, distributed algorithms for optimization and control. 
A comprehensive overview of algorithms in cooperative and competitive MARL with a focus on theoretical analyses of their convergence and complexity was presented in \cite{zhang2021multi}. 

Cooperative MARL tasks are one instance of the setup described above, wherein all agents share the same global reward. 
The authors of \cite{sunehag2018value} introduced value decomposition networks that decomposed a centralized value into a  
sum of individual agent values to assess contributions of individual agents to a shared global reward. 
An additional assumption on monotonicity of the centralized value function 
was imposed in QMIX \cite{rashid2018qmix} to assign credit to an agent by enumerating a value for each valid action at a state. 
The action spaces of agents in the above-mentioned works were discrete and finite, and these techniques cannot be easily adapted to settings with continuous action spaces. 
In comparison, we study reward shaping in MARL tasks in environments with continuous action spaces.
 
 

An alternative approach to improve credit assignment is potential-based reward shaping. 
Although this requires prior knowledge of the problem domain, potential-based techniques have been shown to offer 
guarantees on optimality and convergence of policies in both single \cite{ng1999policy} and multi-agent \cite{devlin2011empirical, devlin2011theoretical, lu2011policy} cases. 
The aforementioned works had focused on the use of potential-based methods in environments with discrete action spaces. 
A preliminary version of this paper \cite{xiao2019potential} introduced potential-based techniques to learn stochastic policies in single-agent RL with continuous states and actions. 
In this paper, we adapt and extend methods from \cite{xiao2019potential} to develop a comprehensive framework for providing potential-based shaping advice in both single- and multi-agent RL. 

The authors of \cite{gangwani2020learning} developed a method called iterative relative credit refinement (IRCR), which used a `surrogate objective' to uniformly redistribute a sparse reward along the length of a trajectory in single and multi-agent RL. 
We empirically compare SAM with IRCR, and explain why SAM is able to guide agents to learn policies that result in higher average rewards than IRCR. 


Methods to learn a potential-based reward shaping function have been investigated 
in recent research. 
When domain knowledge was available, the authors of \cite{harutyunyan2015expressing} proposed a temporal-difference learning method to transform a given reward function to potential-based advice. 
In the absence of domain knowledge, graph convolutional networks were used in \cite{klissarov2020reward} to generate potential functions `automatically' by performing message passing between states at which the agent received a reward. 
Different from the above works, an approach to perform credit assignment by conditioning the value function on future events from a trajectory was examined by the authors of \cite{mesnard2020counterfactual}. 
We note that these works have focused on the single-agent case. 
An analogous approach for the multi-agent case, while interesting, is beyond the scope of the present paper, and remains a promising direction for future research. 
%
%

%% file: Background.tex
\section{Background}\label{Background}

This section 
presents some preliminary material on single- and multi-agent RL.
\subsection{Single-agent Reinforcement Learning}

An MDP \cite{puterman2014markov} is a tuple $(S,A,\mathbb{T},\rho_0, R)$. 
$S$ is the set of states, $A$ the set of actions, 
$\mathbb{T}:S \times A \times S \rightarrow [0,1]$ encodes $\mathbb{P}(s_{t+1}|s_t,a_t)$, the probability of transition to $s_{t+1}$, given current state $s_t$ and action $a_t$. 
$\rho_0$ is a probability distribution over the initial states. 
$R : S \times A \rightarrow \mathbb{R}$ denotes the reward that the agent receives when transitioning from $s_t$ while taking action $a_t$. 
In this paper, $R < \infty$. 

The goal for an RL agent \cite{sutton2018reinforcement} is to learn a \emph{policy} $\pi$, 
in order to maximize 
$J:=\mathbb{E}_{\tau \sim \pi}[\sum_{t=0}^{\infty}\gamma^t R(s_t,a_t)]$. 
Here, $\gamma$ is a discounting factor, and the expectation is taken over the trajectory $\tau=(s_0,a_0,r_0,s_1,\dots)$ induced by policy $\pi$. 
If $\pi: S \rightarrow A$, the policy is \emph{deterministic}. 
On the other hand, a randomized policy returns a probability distribution over the set of actions, and is denoted $\pi: S \times A \rightarrow [0,1]$. 

The value of a state-action pair $(s,a)$ following policy $\pi$ is represented by the \emph{Q-function}, written $Q^{\pi}(s,a) = \mathbb{E}_{\tau \sim \pi}[\sum_{t=0}^{\infty}\gamma^t R(s_t,a_t)|s_0=s,a_0=a]$. 
The Q-function allows us to calculate the state value $V^{\pi}(s) = \mathbb{E}_{a \sim \pi}[Q^{\pi}(s,a)]$. 
The advantage of a particular action $a$, over other actions at a state $s$ is defined by $A^{\pi}(s,a) := Q^{\pi}(s,a)-V^{\pi}(s)$.
%
%
%

\subsection{Stochastic Games for Multi-agent Reinforcement Learning}

A \emph{stochastic game} with $n$ players is a tuple $\mathcal{G} = (S, A^1, \dots, A^n, \mathbb{T}, R^1, \dots, R^n, O^1,\dots,O^n,\rho_0, \gamma)$. $S$ is the set of states, $A^i$ is the action set of player $i$, 
$\mathbb{T}: S \times A^1 \times \dots \times A^n \times S \rightarrow [0,1]$ encodes $\mathbb{P}(s_{t+1}|s_t, a^1_t, \dots, a^n_t)$, the probability of transition to state $s_{t+1}$ from $s_t$, given the respective player actions. 
$R^i: S \times A^1 \times \dots \times A^n \rightarrow \mathbb{R}$ is the reward obtained by agent $i$ when transiting from $s_t$ while each player takes action $a^i_t$. 
$O^i$ is the set of observations for agent $i$. 
At every state, each agent receives an observation correlated with the state: $o^i : S \rightarrow O^i$. 
$\rho_0$ is a distribution over the initial states, and $\gamma \in [0,1]$ is a discounting factor. 

A \emph{policy} for agent $i$ is a distribution over actions, defined by $\pi^i: O^i \times A^i \rightarrow [0,1]$. 
Let $\bm{\pi}:=\{\pi^1, \dots, \pi^n\}$ and $s:=(s^1,\dots,s^n)$. 
Following \cite{lowe2017multi}, in the simplest case, $s^i = o^i$ for each agent $i$, and we use this for the remainder of the paper. 
Additional information about states of agents can be included since we compute \emph{centralized} value functions. 
Let $V_i ^{\bm{\pi}} (s)=V_i (s, \pi^1, \dots, \pi^n):= \mathbb{E}_{\bm{\pi}} [\sum_t \gamma^t R^i_t | s_0 = s, \bm{\pi}]$ 
and $Q_i ^{\bm{\pi}} (s, a^1,\dots,a^n):=R^i + \gamma \mathbb{E}_{s'}[V_i ^{\bm{\pi}} (s')]$ where $V_i ^{\bm{\pi}} (s)= \mathbb{E}_{\{a^i \sim \pi^i\}_{i=1}^n}[Q_i ^{\bm{\pi}} (s, a^1,\dots,a^n)]$.  

\subsection{Reward Shaping in Reinforcement Learning}

Reward shaping methods augment the environment reward $R$ with an additional reward $F \in \mathbb{R}$, $F< \infty$. 
This changes 
the structure of the original MDP $M(=(S,A,\mathbb{T},\rho_0, R))$ to $M'=(S,A,\mathbb{T},\rho_0, R+F)$. 
The goal is to choose $F$ so that an optimal policy for $M'$, $\pi^{*}_{M'}$, is also optimal for the original MDP $M$. 
%
\emph{Potential-based reward shaping} (PBRS) schemes were shown to be able to preserve the optimality of deterministic policies in \cite{ng1999policy}.
PBRS was used in model-based RL in [20], in episodic RL in \cite{grzes2017reward}, and was extended to planning in partially observable domains in \cite{eck2016potential}.
These works focused on the finite-horizon case. 
In comparison, we consider the infinite horizon, discounted cost setting.

The additional reward $F$ in PBRS is defined as a difference of \emph{potentials}, $\phi(\cdot)$. 
Specifically, $F(s_t,a_t,s_{t+1}) := \gamma \phi(s_{t+1}) - \phi(s_t)$. 
Then, the Q-function, $Q^{*}_{M}(s,a)$, of the optimal greedy policy for $M$ and the optimal Q-function $Q^{*}_{M'}(s,a)$ for $M'$ are related by: 
$Q^{*}_{M'}(s,a)= Q^{*}_{M}(s,a) - \phi(s)$.
Therefore, the optimal greedy policy is not changed \cite{ng1999policy, devlin2012dynamic}, since:
\begin{align*}
	&\pi^{*}_{M'}(s) \in \argmax_{a\in A}~Q^{*}_{M'}(s,a) \\
	&\qquad = \argmax_{a\in A}~\big(Q^{*}_{M}(s,a) - \phi(s)\big) = \argmax_{a\in A}~Q^{*}_{M}(s,a).
\end{align*}

The authors of \cite{wiewiora2003principled} augmented $\phi(s)$ to include action $a$ as an argument and termed this \emph{potential-based advice} (PBA). 
They defined two forms of PBA, \emph{look-ahead PBA} and \emph{look-back PBA}, respectively defined by:
\begin{align}
	F(s_{t},a_{t},s_{t+1},a_{t+1}) &= \gamma \phi(s_{t+1},a_{t+1}) - \phi(s_{t},a_{t})\label{lookaheadPBA}\\
	F(s_{t},a_{t},s_{t-1},a_{t-1}) &=  \phi(s_{t},a_{t}) - {\gamma}^{-1}\phi(s_{t-1},a_{t-1}).\label{lookbackPBA}
\end{align}

For the look-ahead PBA scheme, the state-action value function for $M$ following policy $\pi$ is given by: 
\begin{align}\label{PBAq}
	Q^{\pi}_{M}(s,a) =  Q^{\pi}_{M'}(s,a)+\phi(s,a).
\end{align} 
The optimal greedy policy for $M$ can be recovered from the optimal state-action value function for $M'$ using the fact \cite{wiewiora2003principled}:
\begin{align}\label{PBAp}
	\pi^*_{M}(s_t) &\in \argmax_{a \in A} \big(Q^{*}_{M'}(s_t,a)+\phi(s_t,a)\big).
\end{align}
The optimal greedy policy for $M$ using look-back PBA can be recovered similarly.

%
%
%
%
%
%
In the multi-agent case, the \emph{shaping advice} for an agent at each time is a function of observations and actions of all agents. 
The shaping advice is augmented to the environment reward during training, and can take one of two forms, 
\emph{look-ahead} and \emph{look-back}, respectively given by: 
\begin{align}
	&F^i_t(s_t,a^i_t,a^{-i}_t,s_{t+1},a^i_{t+1},a^{-i}_{t+1})\label{LAPBA}\\ &\qquad \qquad :=\gamma \phi_i(s_{t+1},a^i_{t+1},a^{-i}_{t+1}) - \phi_i(s_t,a^i_t,a^{-i}_t)\nonumber 
\end{align}
\begin{align}
	&F^i_t(s_t,a^i_t,a^{-i}_t,s_{t-1},a^i_{t-1}, a^{-i}_{t-1})\label{LBPBA}\\&\qquad \qquad:=\phi_i(s_t,a^i_t,a^{-i}_t) - \gamma^{-1} \phi_i(s_{t-1},a^i_{t-1}, a^{-i}_{t-1})\nonumber
\end{align} 
We will denote by $\mathcal{G}'$ the $n$ player stochastic game that is identical to $\mathcal{G}$, but with rewards $R'^{i}:=R^i + F^i$ for each $i$. 

When the value of the potential function is identical for all actions in a particular state, we will term this \textbf{uniform advice}.  On the other hand, when the value of the potential function depends on the action taken in a state, we will term this \textbf{nonuniform advice}. 
We will explicitly distinguish between uniform and nonuniform variants of shaping advice in the single and multi-agent settings subsequently.

The shaping advice is a heuristic that uses knowledge of the environment and	 task, along with information available to the agent \cite{gupta2017cooperative}. 
For example, in the particle world tasks that we study, each agent has access to positions of other agents and of landmarks, relative to itself. 
This is used to design shaping advice for individual agents at each time step. 

In this paper, we develop a framework for incorporating shaping advice in RL environments with continuous action spaces. 
Moreover, in the multi-agent case, different from prior works that emphasize learning joint policies for agents, we learn decentralized policies for each agent.

%% file: Methods.tex
\section{Shaping Advice in Single-Agent RL}\label{Methods1}

This section presents our results when potential-based shaping advice is used to learn stochastic policies in single-agent RL. 
We term this \textbf{\emph{Shaping Advice in Single-agent RL (SAS)}}. 
This generalizes and extends the use of potential-based methods in the literature which have hitherto focused on augmenting value-based methods to learn optimal deterministic policies. 
We consider two variants- (i) \emph{SAS-Uniform} when the advice for all actions at a particular state is the same, and (ii) \emph{SAS-NonUniform}, when a higher weight might be given to some `good' actions over others at each state. 
We first prove that the optimality of stochastic policies is preserved when using \emph{SAS-Uniform}. 
Then, we describe an approach to integrate \emph{SAS-NonUniform} in to policy-gradient algorithms to enable effective learning of stochastic policies in single-agent RL. 

\subsection{Uniform Advice}\label{PBRSSection}

The following result shows that \emph{SAS-Uniform} preserves optimality even when the optimal policy is stochastic. 
%
%
%
\begin{prop}\label{PBRSResult}
Let 	${\pi}_M^*$ denote the optimal policy for an MDP $M$, and suppose that ${\pi}_M^*$ is a stochastic policy. 
Let ${\pi}_{M'}^*$ denote the optimal policy for the MDP $M'$ whose reward is augmented by $F:=\gamma\phi(s_{t+1})-\phi(s_t)$. 
Then, \emph{SAS-Uniform} preserves the optimality of stochastic policies- i.e., ${\pi}_M^* = {\pi}_{M'}^*$. 
\end{prop}

\begin{proof}
	The goal in the original MDP $M$ was to find a policy $\pi$ in order to maximize:
	\begin{align}\label{PBRS_M}
		{\pi}_M^* &= \argmax_{\pi}   \mathbb{E}_{\tau \sim \pi}\left[\sum_{t=0}^{\infty}\gamma^t R(s_t,a_t)\right].
	\end{align}
	
	When using \emph{SAS-Uniform}, the goal is to determine a policy so that:
	\begin{align}
		&{\pi}_{M'}^* = \argmax_{\pi} \mathbb{E}_{\tau \sim \pi}\big[\sum_{t=0}^{\infty}\gamma^t \big(R(s_t,a_t)+ F(s_t,a_t,s_{t+1},a_{t+1})\big)\big] \nonumber\\
		&= \argmax_{\pi}  \mathbb{E}_{\tau \sim \pi}\big[\sum_{t=0}^{\infty}\gamma^t \big(R(s_t,a_t)+\gamma\phi(s_{t+1})-\phi(s_t)\big)\big] \nonumber\\
		&= \argmax_{\pi} \bigg[\mathbb{E}_{\tau \sim \pi}\big[\sum_{t=0}^{\infty}\gamma^t R(s_t,a_t)\big]-\mathbb{E}_{\tau \sim \pi}\big[\phi(s_0)\big]\bigg]\nonumber\\
		&=\argmax_{\pi}~ \mathbb{E}_{\tau \sim \pi}\big[\sum_{t=0}^{\infty}\gamma^t R(s_t,a_t)\big]-\int_{s}\rho_0(s)\phi(s)\text{d}s.\label{PBRS_M'}
	\end{align}
	The last term in Equation (\ref{PBRS_M'}) is constant, and does not affect the identity of the maximizing policy of (\ref{PBRS_M}).
\end{proof}

\subsection{Nonuniform Advice}\label{PBASection}

Although \emph{SAS-Uniform} preserves optimality of policies in several settings, it suffers from the drawback of being unable to encode richer information, such as desired relations between states and actions. 
The authors of \cite{wiewiora2003principled} proposed \emph{potential-based nonuniform advice}, a scheme that augmented the potential function by including actions as an argument together with states. 
In this part, we show that when using \emph{SAS-NonUniform}, recovering the optimal policy can be difficult if the optimal policy is stochastic. 
To overcome this barrier, we propose a novel way to impart prior information about the environment in order to use \emph{SAS-NonUniform} to learn a stochastic policy.
\subsubsection{Stochastic policy learning with nonuniform advice}

Assume that we can compute $Q^{*}_{M}(s,a)$, the optimal value for state-action pair $(s,a)$ in MDP $M$. 
The optimal stochastic policy for $M$ is $\pi^*_M = \argmax_{\tau \sim \pi}\mathbb{E}_{\pi}\big[Q^{*}_{M}(s,a)\big]$. 
From Eqn. (\ref{PBAq}), the optimal stochastic policy for the modified MDP $M'$ 
is given by $\pi^*_{M'} = \argmax_{\pi}\mathbb{E}_{\tau \sim \pi}\big[Q^{*}_{M}(s,a)-\phi(s,a)\big]$. 
Without loss of generality, $\pi^*_M \neq \pi^*_{M'}$. 
If the optimal policy is deterministic, then the policy for $M$ can be recovered easily from that for $M'$ using Eqn. (\ref{PBAp}). 
However, for stochastic optimal policies, we will need to average over trajectories of the MDP, which makes it difficult to recover the optimal policy for $M$ from that of $M'$. 

In what follows, we will propose a novel way to take advantage of \emph{SAS-NonUniform} in the policy gradient framework in order to directly learn a stochastic policy.
\subsubsection{Imparting nonuniform advice in policy gradient}

Let $J_M(\theta)$ denote the value of a parameterized policy $\pi_{\theta}$ in MDP $M$. 
That is, $J_M(\theta) = \mathbb{E}_{\tau \sim \pi_{\theta}}\left[\sum_{t=0}^{\infty}\gamma^t R(s_t,a_t)\right]$. 
Following the policy gradient theorem \cite{sutton2018reinforcement}, and defining $G(s_t,a_t):=\sum_{i=t}^{i=\infty}\gamma^{i-t}r_i$, the gradient of $J(\theta)$ with respect to $\theta$ is:
\begin{align}\label{REINFORCE}
	\nabla_{\theta}J_M(\theta) = \mathbb{E}_{\tau \sim \pi_{\theta}}\big[G(s_t,a_t)\nabla_{\theta}\log\pi_{\theta}(a_t|s_t)\big].
\end{align}
Then, $\mathbb{E}_{\tau \sim\pi_{\theta}}\big[G(s_t,a_t)\big]=Q^{\pi_{\theta}}(s_t,a_t)$. 

REINFORCE \cite{sutton2018reinforcement} is a policy gradient method that uses Monte Carlo simulation to learn $\theta$, where the parameter update is performed only at the end of an episode (a trajectory of length $T$). 
If we apply the look-ahead variant of \emph{SAS-NonUniform} as in Equation (\ref{lookaheadPBA}) along with REINFORCE, then the total return from time $t$ is given by:
\begin{align}
	\begin{split}
		G^{a}(s_t,a_t)&=\sum_{i=t}^{i=T}\gamma^{i-t}r_i+\gamma^{T-t}\phi(s_T,a_T)-\phi(s_t,a_t) \\
		&=G(s_t,a_t)+\gamma^{T-t}\phi(s_T,a_T)-\phi(s_t,a_t).
	\end{split}
\end{align}

Notice that if $G^{a}(s_t,a_t)$ is used in Eqn. (\ref{REINFORCE}) instead of $G(s_t,a_t)$, then the policy gradient is biased. 
One way to resolve the problem is to add the difference $-\gamma^{T-t}\phi(s_T,a_T)+\phi(s_t,a_t)$ to $G^{a}(s_t,a_t)$. 
However, this makes the learning process identical to the original REINFORCE and nonuniform advice is not used. 
When using nonuniform advice in a policy gradient setup, it is important to add the term $\phi(s,a)$ so that the policy gradient is unbiased, and also leverage the advantage that nonuniform advice offers during learning.
%
\subsection{Analysis and Algorithm} 

To integrate \emph{SAS-NonUniform} with policy gradient-based techniques, we turn to temporal difference (TD) methods. 
TD methods update estimates of the accumulated return based in part on other learned estimates, before the end of an episode.
A popular TD-based policy gradient method is the actor-critic framework \cite{sutton2018reinforcement}. 
In this setup, after performing action $a_t$ at step $t$, the accumulated return $G(s_t,a_t)$ is estimated by $Q_M(s_t,a_t)$ which, in turn, is estimated by $r_t+\gamma V_M(s_{t+1})$. 
It should be noted that the estimates are unbiased.

When the reward is augmented with look-ahead \emph{SAS-NonUniform}, the accumulated return is changed to $Q_{M'}(s_t,a_t)$, which is estimated by $r_t+\gamma\phi(s_{t+1},a_{t+1})-\phi(s_t,a_t)+\gamma V_{M'}(s_{t+1})$. 
From Eqn. (\ref{PBAq}), at steady state, $Q_M(s_t,a_t)-Q_{M'}(s_t,a_t)=\phi(s_t,a_t)$. 
Intuitively, to keep policy gradient unbiased when augmented with look-ahead nonuniform advice, we can add $\phi(s_t,a_t)$ at each training step. 
In other words, we can use $r_t+\gamma\phi(s_{t+1},a_{t+1})+\gamma V_{M'}(s_{t+1})$ as the estimated return. 
It should be noted that before the policy reaches steady state, adding $\phi(s_t,a_t)$ at each time step will not cancel out the effect of nonuniform advice. 
This is unlike in REINFORCE, where the addition of this term negates the effect of using nonuniform advice. 
In the advantage actor-critic, an advantage term is used instead of the Q-function in order to reduce the variance of the estimated policy gradient. 
In this case also, the potential term $\phi(s_t,a_t)$ can be added in order to keep the policy gradient unbiased. 
\begin{algorithm}
	\small
	\caption{SAS: Shaping Advice in Single-agent RL}
	\begin{algorithmic}[1] \label{Algo1}
		\renewcommand{\algorithmicrequire}{\textbf{Input:}}
%
%
%
%
		\REQUIRE Differentiable policy and value functions $\pi_{\theta}(a|s)$, $V_{\omega}(s)$; shaping advice $\phi(s,a)$; maximum episode length $T_{max}$.
		
		\textbf{Initialization}: \\
		policy parameter $\theta$, value parameter $\omega$, learning rate $\alpha^{\theta}$ and $\alpha^{\omega}$, discount factor $\gamma$.
		\STATE{$T = 0$}	
		\REPEAT
		\STATE initialize state $s_0$, $t \leftarrow 0$
		\REPEAT
		\STATE Sample action $a_t \sim \pi_{\theta}(\cdot|s_t)$
		\STATE Take action $a_t$, observe reward $r_t$, next state $s_{t+1}$
		\STATE $R= \begin{cases}
			0, & \text{if }
			\begin{aligned}[t]
				s_{t+1} \text{ is a terminal state },
			\end{aligned}
			\\
			V_{\omega}(s_{t+1}), & \text{otherwise.}
		\end{cases}$
		
		\IF{use look-ahead advice} 
		\STATE $\delta_t=r_t + \gamma\phi(s_{t+1},a_{t+1})-\phi(s_t,a_t)+\gamma R - V_{\omega}(s_t)$ 
		\STATE Update $\theta \leftarrow \theta + \alpha^{\theta} \big(\delta_t+\phi(s_t,a_t)\big)\nabla_{\theta}\log\pi_{\theta}(a_t|s_t)$
		
		\ELSE 
		\STATE $\delta_t = r_t + \phi(s_{t},a_{t})-\gamma^{-1}\phi(s_{t-1},a_{t-1})+\gamma R - V_{\omega}(s_t)$
		\STATE Update $\theta \leftarrow \theta + \alpha^{\theta} \delta_t\nabla_{\theta}\log\pi_{\theta}(a_t|s_t)$
		\ENDIF
		
		\STATE Update $\omega \leftarrow \omega - \alpha^{\omega} \delta_t\nabla_{\omega}V_{\omega}(s_t)$
		
		\UNTIL{$s_{t+1}$ is a terminal state}
		\STATE $T \leftarrow T+1$
		\UNTIL{$T>T_{max}$}	
	\end{algorithmic}
\end{algorithm}

A procedure for augmenting the advantage actor-critic with \emph{SAS-NonUniform} is presented in Algorithm \ref{Algo1}. 
$\alpha^{\theta}$ and $\alpha^{\omega}$ denote learning rates for the actor and critic respectively. 
When applying look-ahead nonuniform advice, at training step $t$, parameter $\omega$ of the critic $V_{\omega}(s)$ is updated as follows:
\begin{align}
	\delta^a_t &= r_t + \gamma\phi(s_{t+1},a_{t+1})-\phi(s_t,a_t)+\gamma V_{\omega}(s_{t+1}) - V_{\omega}(s_t)\nonumber \\
	\omega &= \omega - \alpha^{\omega} \delta^a_t\nabla_{\omega}V_{\omega}(s_t),\nonumber
\end{align}
where $\delta^a_t$ is the estimation error of the state value after receiving new reward $[r_t + \gamma\phi(s_{t+1},a_{t+1})-\phi(s_t,a_t)]$ at step $t$.
To ensure an unbiased estimate of the policy gradient, the potential term $\phi(s_t,a_t)$ is added while updating $\theta$ as \cite{xiao2019potential}:
\begin{align}
	\theta = \theta + \alpha^{\theta} \big(\delta^a_t+\phi(s_t,a_t)\big)\nabla_{\theta}\log\pi_{\theta}(a_t|s_t).\nonumber
\end{align}

A similar method can be used when learning with look-back nonuniform advice. 
In this case, the critic and the policy parameter are updated as follows:
\begin{align}
	\delta^b_t &= r_t + \phi(s_{t},a_{t})-\gamma^{-1}\phi(s_{t-1},a_{t-1})+\gamma V_{\omega}(s_{t+1}) - V_{\omega}(s_t)\nonumber\\
	\omega &= \omega - \alpha^{\omega} \delta^b_t\nabla_{\omega}V_{\omega}(s_t),\nonumber\\
	\theta &= \theta + \alpha \big(\delta^b_t+\gamma^{-1}\mathbb{E}\big[\phi(s_{t-1},a_{t-1})|s_t\big]\big)\nabla_{\theta}\log\pi_{\theta}(a_t|s_t)\label{addback}
\end{align}

In the above case, the potential term need not be added to ensure an unbiased estimate. 
Then, the policy parameter update becomes:
\begin{align}\label{noaddback}
	\theta = \theta + \alpha \delta^b_t\nabla_{\theta}\log\pi_{\theta}(a_t|s_t),
\end{align}
which is exactly the policy update of the advantage actor-critic. 
This is formally stated in Proposition \ref{PBAprop} 
\begin{prop} \label{PBAprop}
	When the actor-critic is augmented with look-back variant of \emph{SAS-NonUniform}, Equations (\ref{addback}) and (\ref{noaddback}) are equal in the sense of expectation. That is
	\begin{align}   
		\mathbb{E}_{(s_t,a_t) \sim \rho^{\pi_\theta}}\big[\big(\delta^b_t+&\gamma^{-1}\mathbb{E}\big[\phi(s_{t-1},a_{t-1})|s_t\big]\big)\nabla_{\theta}\log\pi_{\theta}(a_t|s_t)\big] \nonumber\\
		=	\quad&\mathbb{E}_{(s_t,a_t) \sim \rho^{\pi_\theta}}\big[\delta^b_t\nabla_{\theta}\log\pi_{\theta}(a_t|s_t)\big],
	\end{align} 
	where $\rho^{\pi_\theta}$ is the distribution induced by the policy $\pi_\theta$.
\end{prop}
\begin{proof}
	It is equivalent to show that:
	\begin{align}   
		\mathbb{E}_{(s_t,a_t) \sim \rho^{\pi_\theta}}\big[\mathbb{E}\big[\phi(s_{t-1},a_{t-1})|s_t\big]\nabla_{\theta}\log\pi_{\theta}(a_t|s_t)\big] = 0.
	\end{align} 
	The inner expectation $\mathbb{E}\big[\phi(s_{t-1},a_{t-1})|s_t\big]$ is a function of $s_t$, policy $\pi_{\theta}$, and transition probability $\mathbb{T}$. 
	Denoting this expectation by $f(s_t,\pi_{\theta},\mathbb{T})$, we obtain:
	\begin{align}   
		&\mathbb{E}_{(s_t,a_t) \sim \rho^{\pi_\theta}}\big[f(s_t,\pi_{\theta},\mathbb{T})\nabla_{\theta}\log\pi_{\theta}(a_t|s_t)\big]\nonumber \\= 
		&\mathbb{E}_{s_t \sim
			\rho^{\pi_\theta}}\bigg[\mathbb{E}_{a_t\sim \pi_\theta}\big[f(s_t,\pi_{\theta},\mathbb{T})\nabla_{\theta}\log\pi_{\theta}(a_t|s_t)\big]\bigg]\nonumber \\= 
		&\mathbb{E}_{s_t \sim
			\rho^{\pi_\theta}}\bigg[\int_{A}\pi_{\theta}(a_t|s_t)f(s_t,\pi_{\theta},\mathbb{T})\frac{\nabla_{\theta}\pi_{\theta}(a_t|s_t)}{\pi_{\theta}(a_t|s_t)}\text{d}a\bigg] \nonumber\\= 
		&\mathbb{E}_{s_t \sim
			\rho^{\pi_\theta}}\bigg[f(s_t,\pi_{\theta},\mathbb{T})\nabla_{\theta}\int_{A}\pi_{\theta}(a_t|s_t)\text{d}a\bigg] = 0.
	\end{align}
	The last equality follows from the fact that the integral evaluates to $1$, and its gradient is $0$.
\end{proof}

\begin{rem}
	Look-back nonuniform advice could result in better performance compared to look-ahead nonuniform advice since look-back nonuniform advice does not involve estimating a future action. 
\end{rem}

\section{Shaping Advice in Multi-Agent RL}\label{Methods2}
\begin{figure}[!h]
	\centering
	\includegraphics[width=3.05 in]{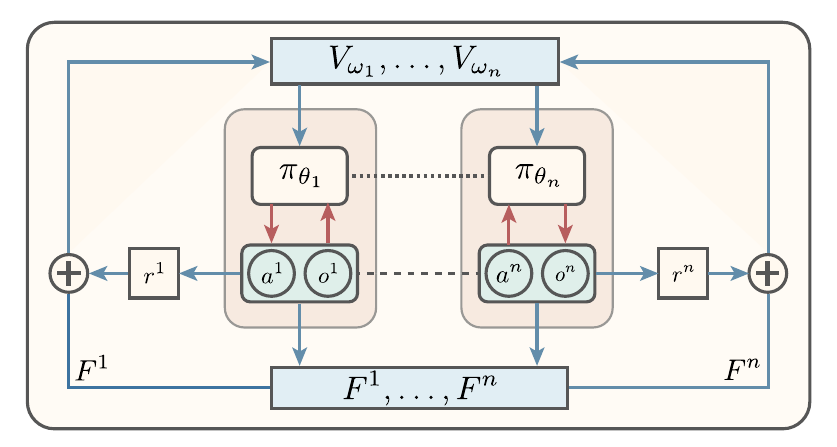} 
	\caption{Schematic of SAM. A centralized critic estimates value functions $V_{\omega_1},\dots,V_{\omega_n}$. Actions for an agent $i$ are sampled from its policy $\pi_{\theta_i}$ in a decentralized manner. Actions and observations of all agents are used to determine \emph{shaping advice} $F^1,\dots,F^n$. The advice $F^i$ is augmented to the reward $r^i$ from the RL environment. The workflow shown by blue arrows in the outer box is required only during training. During execution, only the workflow shown by the red arrows inside the inner boxes is needed.}\label{SAMARITANSchematic}
\end{figure}

This section introduces shaping advice in multi-agent reinforcement learning (\emph{SAM}).
The goal of SAM is to augment \emph{shaping advice} to the reward supplied by the MARL environment to provide immediate feedback to agents on their actions. 
SAM uses the CTDE paradigm wherein agents share parameters with each other during the training phase, but execute decentralized policies using only their own observations at test-time. 
%
Figure \ref{SAMARITANSchematic} shows a schematic of SAM. 
We subsequently detail how the shaping advice is provided to the agents, and analyze the optimality and convergence of policies when using SAM.

\subsection{Centralized Critic}

SAM uses a centralized critic during the training phase. 
Information about observations and actions of all agents is used to learn a decentralized policy for each agent. 
One way to do this is by using an actor-critic framework, which combines policy gradients with \emph{temporal difference (TD)} techniques. 

At time $t$, the joint action $(a^1_t,\dots,a^n_t)$ is used to estimate the accumulated return for each agent $i$ as $r^i_t + \gamma V^i (s_{t+1})$. 
This quantity is called the \emph{TD-target}. 
Subtracting $V^i(s_t)$ from the TD-target gives the \emph{TD-error}, which is an unbiased estimate of the agent advantage \cite{sutton2018reinforcement}. 
Each actor can then be updated following a gradient based on this TD-error. 

We learn a separate critic for each agent like in \cite{lowe2017multi}. However, the learning process can be affected when rewards provided by the environment are sparse. 
SAM uses a \emph{potential-based} heuristic as shaping advice that is augmented to the reward received from the environment. 
This resulting reward is less sparse and can be used by the agents to learn policies.

\subsection{Shaping Advice in Multi-Agent Actor-Critic}
%

We describe how to augment shaping advice to the multi-agent policy gradient to assign credit. 
We use the actor-critic framework with a centralized critic and decentralized actors. 

For an agent $i$, shaping advice $F^i$ is augmented to the environment reward $r^i$ at each time step. 
$F^i$ is specified by a difference of potentials (Eqn. (\ref{LAPBA}) or (\ref{LBPBA})). 
The centralized critic allows using observations and actions of all agents to specify $F^i$. 
Using look-ahead advice, $Q$-values in the modified game $\mathcal{G}'$ with reward $R^i+F^i$ and original game $\mathcal{G}$ with reward $R^i$ are related as \cite{wiewiora2003principled}:
\begin{align}
[Q^{\bm{\pi_\theta}}_i(s_t,a^i_t,a^{-i}_t)]_{\mathcal{G}}&=[Q^{\bm{\pi_\theta}}_i(s_t,a^i_t,a^{-i}_t)]_{\mathcal{G}'}\nonumber\\&\qquad \qquad+\phi_i(s_t,a^i_t,a^{-i}_t) \label{RelnGG'}
\end{align}

The accumulated return in $\mathcal{G}'$ for agent $i$ is then estimated by $r^i_t+ \gamma \phi_i(s_{t+1},a^i_{t+1},a^{-i}_{t+1}) - \phi_i(s_t,a^i_t,a^{-i}_t)+ \gamma V^i (s_{t+1})$. 
From Eqn. (\ref{RelnGG'}), we can add $\phi_i(s_t,a^i_t,a^{-i}_t)$ to the TD-target in $\mathcal{G}'$ at each time $t$ to keep the policy gradient unbiased in $\mathcal{G}$. 

Let the critic and actor in SAM for agent $i$ be respectively parameterized by $\omega_i$ and $\theta_i$.
When the actor is updated at a slower rate than critic, the asymptotic behavior of the critic can be analyzed by keeping the actor fixed using \emph{two time-scale stochastic approximation} methods \cite{borkar2009stochastic}. 
%
For agent $i$, the TD-error at time $t$ is given by:
\begin{align}
\delta^i_t&:=r^i_t+F^i_t+\gamma V_{\omega_i} (s_{t+1}) - V_{\omega_i} (s_t). \label{TDErrori}
\end{align}
The update of the critic can be expressed as a first-order ordinary differential equation (ODE) in $\omega_i$, given by: 
\begin{align}
\dot{\omega}_i &= \mathbb{E}_{\bm{\pi_\theta}}[\delta^i \nabla_{\omega_i} V_{\omega_i} (s_t)] \label{CriticODE}
\end{align}
Under an appropriate parameterization of the value function, this ODE will converge to an asymptotically stable equilibrium, denoted $\omega_i (\bm{\theta})$. 
At this equilibrium, the TD-error for agent $i$ is $\delta_{t,\omega_i(\bm{\theta})}^i = r^i_t+F^i_t +\gamma V_{\omega_i(\bm{\theta})} (s_{t+1}) - V_{\omega_i(\bm{\theta})} (s_t)$. 

The update of the actor can then be determined by solving a first order ODE in $\theta_i$. 
With look-ahead advice, a term corresponding to the shaping advice at time $t$ will have to be added to ensure an unbiased policy gradient (Eqn. (\ref{RelnGG'})). 
This ODE can be written as: 
\begin{align}
\dot{\theta}_i &= \mathbb{E}_{\bm{\pi_\theta}}[(\delta_{t,\omega_i(\bm{\theta})}^i +\phi_i(s_t,a^i_t,a^{-i}_t))  \nabla_{\theta_i} \log~\pi_{\theta_i}(a_t^i|o^i_t)] \label{ActorODE}
\end{align} 
A potential term will not have to be added to ensure an unbiased policy gradient when utilizing look-back advice. This insight follows from Proposition 3 in \cite{xiao2019potential} since we consider decentralized policies. 
%
\subsection{Analysis and Algorithm}

In this part, we present a proof of the convergence of the actor and critic parameters when learning with shaping advice. 
We also demonstrate that convergence of policy gradients and values when using SAM implies convergence of these quantities in the absence of SAM. 
This will guarantee that policies learned in the modified stochastic game $\mathcal{G}'$ will be locally optimal in the original game $\mathcal{G}$. 

For agent $i$, the update dynamics of the critic can be expressed by the ODE in Eqn. (\ref{CriticODE}). 
Assuming parameterization of $V(s)$ over a linear family, this ODE will converge to an asymptotically stable equilibrium \cite{borkar2009stochastic}. 
The actor update is then given by the ODE in Eqn. (\ref{ActorODE}). 
The parameters associated with the critics are assumed to be updated on a faster timescale than those of the actors. 
Then, the behaviors of the actor and critic can be analyzed separately using two timescale stochastic approximation techniques \cite{borkar2009stochastic}. 

\begin{ass}
	We make the following assumptions:
	\begin{enumerate}
		\item At any time $t$, an agent is aware of the actions taken by all other agents. Rewards received by the agents at each time step are uniformly bounded.
		\item The Markov chain induced by the agent policies is irreducible and aperiodic. 
		\item For each agent $i$, the update of its policy parameter $\theta_i$ includes a projection operator $\Gamma_i$, which projects $\theta_i$ onto a compact set $\Theta_i$. We assume that $\Theta_i$ includes a stationary point of $\nabla_{\theta_i} J_i(\bm{\theta})$ for each $i$. 
		\item For each agent $i$, its value function is parameterized by a linear family. That is, $V_{\omega_i} (s) = \Phi_i \omega_i$, where $\Phi_i$ is a known, full-rank feature matrix for each $i$. 
		\item For each agent $i$, the TD-error at each time $t$ and the gradients $\nabla_{\omega_i} V_{\omega_i}(s)$ are bounded, and the gradients 
		$\nabla_{\theta_i} \log ~\pi_{\theta_i} (\cdot|s_t)$ are Lipschitz with bounded norm. 
		\item The learning rates satisfy $\sum_t \alpha^\theta_t = \sum_t \alpha^\omega_t = \infty$, $\sum_t[(\alpha^\theta_t)^2 + (\alpha^\omega_t)^2] < \infty$, $\lim_{t \rightarrow \infty} \frac{\alpha^\theta_t}{\alpha^\omega_t}=0$.
	\end{enumerate}
\end{ass}
We first state a useful result from \cite{kushner2012stochastic}.
\begin{lm}[\cite{kushner2012stochastic}]\label{KushnerClarkLemma}
	Let $\Gamma:\mathbb{R}^k \rightarrow \mathbb{R}^k$ be a projection onto a compact set $K \subset \mathbb{R}^k$. 
	Define 
	\begin{align}
	\hat{\Gamma}(h(x)):&=\lim_{\epsilon \downarrow 0} \frac{\Gamma(x+\epsilon h(x))-x}{\epsilon} \nonumber
	\end{align}
	for $x \in K$ and $h:\mathbb{R}^k \rightarrow \mathbb{R}^k$ continuous on $K$. Consider the update $x_{t+1} = \Gamma (x_t + \alpha_t(h(x_t) + \xi_{t,1}+\xi_{t,2}))$ and its associated ODE $\dot{x} = \hat{\Gamma}(h(x))$. 
	Assume that: 
	
	i) $\{\alpha_t\}$ is such that $\sum_t \alpha_t = \infty$, $\sum_t \alpha_t^2 < \infty$; 
	
	ii) $\{\xi_{t,1}\}$ is such that for all $\epsilon > 0$, $\lim_t \mathbb{P} (\sup_{n \geq t} || \sum_{\tau = t}^n \alpha_{\tau} \xi_{\tau,1}|| \geq \epsilon) = 0$; 
	
	iii) $\{\xi_{t,2}\}$ is an almost surely bounded random sequence, and $\xi_{t,2} \rightarrow 0$ almost surely. 

	Then, if the set of asymptotically stable equilibria of the ODE in $\dot{x}$ is compact, denoted $K_{eq}$, the updates $x_{t+1}$ will converge almost surely to $K_{eq}$. 
\end{lm}
Let $\{\mathcal{F}^\omega_t\}$ be the filtration where $\mathcal{F}^\omega_t:=\sigma(s_{\tau},$ $r^1_{\tau},\dots,r^n_{\tau},\omega_{1_\tau},\dots,\omega_{n_{\tau}}:\tau \leq t)$ is an increasing $\sigma-$algebra generated by iterates of $\omega_i$ up to time $t$. 
We first analyze behavior of the critic when parameters of the actor are fixed. 

\begin{thm}\label{ThmCriticConv}
	For a fixed policy $\bm{\pi_\theta}$, the update $\omega_i \leftarrow \omega_i - \alpha_t^\omega \delta^i_t \nabla_{\omega_i} V_{\omega_i} (s_t)$ converges almost surely to the set of asymptotically stable equilibria of the ODE $\dot{\omega}_i=h_i(\omega_i):=\mathbb{E}_{\bm{\pi_\theta}}[\delta^i_t \nabla_{\omega_i}V_{\omega_i}(s_t)|\mathcal{F}^\omega_t]$. 
\end{thm}

\begin{proof}
	Let $\xi^i_{t,1}:=\delta^i_t \nabla_{\omega_i} V_{\omega_i} (s_t)-\mathbb{E}_{\bm{\pi_\theta}}[\delta^i_t \nabla_{\omega_i}V_{\omega_i}(s_t)|\mathcal{F}^\omega_t]$. 
	Then, the $\omega_i$ update can be written as $\omega_i \leftarrow \omega_i- \alpha_t^\omega[h_i(\omega_i)+\xi^i_{t,1}]$, where  
	$h_i(\omega_i)$ is continuous in $\omega_i$. 
	Since $\delta^i_t$ and $\nabla_{\omega_i} V_{\omega_i}(s)$ are bounded, $\xi^i_{t,1}$ is almost surely bounded. 
	
	Let $M^i_t:= \sum_{\tau = 0}^t \alpha^\omega_{\tau}  \xi^i_{\tau,1}$. 
	Then $\{M^i_t\}$ is a martingale\footnote{A martingale \cite{williams1991probability} is a stochastic process $S_1,S_2,\dots$ that satisfies $\mathbb{E}(|S_n| < \infty)$ and $\mathbb{E}(S_{n+1}|S_1,\dots,S_n) = S_n$ for each $n = 1,2,\dots.$.}
	, and $\sum_t ||M^i_t-M^i_{t-1}||^2 = \sum_t ||\alpha^\omega_t \xi^i_{t,1}||^2 < \infty$ almost surely. 
	Therefore, from the martingale convergence theorem \cite{williams1991probability}, the sequence $\{M^i_t\}$ converges almost surely. 
	Therefore, the conditions in Lemma \ref{KushnerClarkLemma} are satisfied. 
	
	Since $V_{\omega_i} = \Phi_i \omega_i$, with $\Phi_i$ a full-rank matrix, $h_i(\omega_i)$ is a linear function, and the ODE will have a unique equibrium point. 
	This will be an asymptotically stable equilibrium since ODE dynamics will be governed by a matrix of the form $(\gamma T_{\pi} - I)$. 
	Here, $I$ is an identity matrix, and $T_{\pi}$ is a stochastic state-transition matrix under policy $\pi$, whose eigen-values have (strictly) negative real parts \cite{prasad2015two}. 
	Denote this asymptotically stable equilibrium by $\omega_i(\bm{\theta})$. 
\end{proof}
We can now analyze the behavior of the actor, assuming that the critic parameters have converged to an asymptotically stable equilibrium. 
With $\omega_i (\bm{\theta})$ a limit point of the critic update, let $\delta_{t,\omega_i(\bm{\theta})}^i = r^i_t+F^i_t +\gamma V_{\omega_i(\bm{\theta})} (s_{t+1}) - V_{\omega_i(\bm{\theta})} (s_t)$. 
When using look-ahead or look-back advice, define $\tilde{\delta}^i_{t,\omega_i(\bm{\theta})}$ as:
\begin{align}
\begin{split}
&\text{look-ahead:}\quad\tilde{\delta}^i_{t,\omega_i(\bm{\theta})}:= (\delta_{t,\omega_i(\bm{\theta})}^i+\phi_i(s_t,a^i_t,a^{-i}_t)) \\
&\text{look-back:} \quad\tilde{\delta}^i_{t,\omega_i(\bm{\theta})}:= \delta_{t,\omega_i(\bm{\theta})}^i
. 
\end{split}\label{delta}
\end{align}

Let $\{\mathcal{F}^\theta_t\}$ be a filtration where $\mathcal{F}^\theta_t:=\sigma(\bm{\theta}_{\tau}:=[\theta_{1_\tau}\dots \theta_{n_\tau}]:\tau \leq t)$ is an increasing $\sigma-$algebra generated by iterates of $\theta_i$ up to time $t$. 

\begin{thm} \label{ThmActorConv} 
	The update $\theta_i \leftarrow \Gamma_i[\theta_i + \alpha_t^\theta \tilde{\delta}^i_t $ $ \nabla_{\theta_i} \log~\pi_{\theta_i}(a_t^i|o^i_t)]$ converges almost surely to the set of asymptotically stable equilbria of the ODE $\dot{\theta}_i = \hat{\Gamma}_i(h_i(\theta_i))$, where $h_i(\theta_i) = \mathbb{E}_{\bm{\pi_\theta}}[\tilde{\delta}^i_{t,\omega_i(\bm{\theta})}\nabla_{\theta_i} \log~\pi_{\theta_i}(a_t^i|o^i_t)|\mathcal{F}^\theta_t]$.
\end{thm}

\begin{proof}
	Let $\xi^i_{t,1} := \tilde{\delta}^i_t \nabla_{\theta_i} \log \pi_{\theta_i}(a_t^i|o^i_t) - \mathbb{E}_{\bm{\pi_\theta}}[\tilde{\delta}^i_t \nabla_{\theta_i} \log\pi_{\theta_i}(a_t^i|o^i_t) | \mathcal{F}^\theta_t]$ and $\xi^i_{t,2}:= \mathbb{E}_{\bm{\pi_\theta}}[(\tilde{\delta}^i_t - \tilde{\delta}^i_{t,\omega_i(\bm{\theta})})\nabla_{\theta_i} \log \pi_{\theta_i}(a_t^i|o^i_t) | \mathcal{F}^\theta_t]$. 
%
	Then, the update of $\theta_i$ can be written as $\theta_i \leftarrow \theta_i + \alpha^\theta_t [h_i(\theta_i) +\xi^i_{t,1} + \xi^i_{t,2}]$, where 
%
	$h_i(\theta_i)$ is continuous in $\theta_i$. We now need to verify that the conditions in Lemma \ref{KushnerClarkLemma} are satisfied. 
	
	Since the critic parameters converge almost surely to a fixed point, $\tilde{\delta}^i_t - \tilde{\delta}^i_{t,\omega_i(\bm{\theta})} \rightarrow 0$ almost surely. Therefore, $\xi^i_{t,2} \rightarrow 0$ almost surely, verifying 
	Condition iii) in Lemma \ref{KushnerClarkLemma}. 
	
	Since $\tilde{\delta}^i_t$ and $\nabla_{\theta_i} \log \pi_{\theta_i}(a_t^i|o^i_t)$ are bounded, $\xi^i_{t,1}$ is continuous in $\theta_i$ and $\theta_i$ belongs to a compact set, the sequence $\{\xi^i_{t,1}\}$ is bounded almost surely \cite{rudin1964principles}. 
	If $M^i_t:= \sum_{\tau = 0}^t \alpha^\theta_{\tau}  \xi^i_{\tau,1}$, then $\{M^i_t\}$ is a martingale, and $\sum_t ||M^i_t-M^i_{t-1}||^2 = \sum_t ||\alpha^\theta_t \xi^i_{t,1}||^2 < \infty$ almost surely. 
	Then, 
	$\{M^i_t\}$ converges almost surely \cite{williams1991probability}, satisfying  
	Condition ii) of Lemma \ref{KushnerClarkLemma}. 
	Condition i) is true by assumption, completing the proof.
\end{proof} 
Theorems \ref{ThmCriticConv} and \ref{ThmActorConv} demonstrate the convergence of critic and actor parameters in the stochastic game with the shaped reward, $\mathcal{G}'$. 
However, our objective is to provide a guarantee of convergence in the original game $\mathcal{G}$. 
We establish such a guarantee when parameterizations of the value function results in small errors, and policy gradients in $\mathcal{G}'$ are bounded. 

\begin{df}
	For a probability measure $\mu$ on a finite set $\mathcal{M}$, the $\ell_2-$norm of a function $f$ with respect to $\mu$ is defined as $||f||_{\mu}:=\bigg[\int_{\mathcal{M}} |f(X)|^2 d\mu (X)\bigg]^{\frac{1}{2}}=\bigg[\mathbb{E}_{\mu}(|f(X)|^2)\bigg]^{\frac{1}{2}}$. 
\end{df}
\begin{thm}\label{BoundInOrigGame}
	In the stochastic game $\mathcal{G}'$, let  $(\gamma+1)||V_i^{\pi_{\bm{\theta}}}(s) - V_{\omega_i(\bm{\theta})}(s)||_{\pi_{\bm{\theta}}} \leq \mathcal{E}_i(\bm{\theta})$, and let $||\nabla_{\theta_i} \log\pi_{\theta_i}||_{\pi_{\bm{\theta}}} \leq C_i(\bm{\theta})$. 
	Let $(\bm{\theta}^*, \omega(\bm{\theta})^*)$ be the set of limit points of SAM. 
	
	Then, in the original stochastic game $\mathcal{G}$, for each agent $i$, $||\nabla_{\theta_i}J_i(\bm{\theta}^*)||_2 \leq C_i(\bm{\theta}^*)\mathcal{E}_i(\bm{\theta}^*)$.
\end{thm}
\begin{proof}
	Let $\Theta_{i_{eq}}$ denote the set of asymptotically stable equilibria of the ODE in $\theta_i$. 
	Let $\Theta_{eq}:=\Theta_{1_{eq}} \times \dots \times \Theta_{n_{eq}}$. 
	Then, in the set $\Theta_{eq}$, $\dot{\theta}_i = 0$ for each agent $i$. 
	
	Consider a policy $\pi_{\bm{\theta}}$, $\bm{\theta} \in \Theta_{eq}$. 
	In the original game $\mathcal{G}$, 
	\begin{align}
	\nabla_{\theta_i}J_i(\bm{\theta}) &= \mathbb{E}_{\bm{\pi_\theta}} [ \nabla_{\theta_i} \log\pi_{\theta_i} (a^i_t|o^i_t)Q^{\bm{\pi_\theta}}_i(s_t,a^i_t,a^{-i}_t)]\label{MAPolGrad}
	\end{align}
	
	From Equation (\ref{RelnGG'}), $[Q^{\bm{\pi_\theta}}_i(s_t,a^i_t,a^{-i}_t)]_{\mathcal{G}}=[Q^{\bm{\pi_\theta}}_i(s_t,a^i_t,a^{-i}_t)]_{\mathcal{G}'}$ $+\phi_i(s_t,a^i_t,a^{-i}_t)$. 
	Since we use an advantage actor critic, we replace $[Q^{\bm{\pi_\theta}}_i(s_t,a^i_t,a^{-i}_t)]_{\mathcal{G}'}$ with an advantage term, defined as $[Q^{\bm{\pi_\theta}}_i(s_t,a^i_t,a^{-i}_t)]_{\mathcal{G}'}-V^{\bm{\pi_\theta}}_i(s_t)$. 
	Substituting these quantities in Equation (\ref{MAPolGrad}), 
	\begin{align}
	\nabla_{\theta_i}J_i(\bm{\theta}) &= \mathbb{E}_{\bm{\pi_\theta}}[\nabla_{\theta_i} \log\pi_{\theta_i} (a^i_t|o^i_t).\label{MAPolGradG'}\\&\qquad \qquad (r^i_t+F^i_t+\gamma V_i^{\pi_{\bm{\theta}}}(s_{t+1})\nonumber\\&\qquad \qquad-V^{\bm{\pi_\theta}}_i(s_t)+ \phi_i(s_t,a^i_t,a^{-i}_t))]\nonumber
	\end{align}
	
	At equilibrium,  $\dot{\theta}_i= 0$ in Equation (\ref{ActorODE}). 
	Subtracting this from Equation (\ref{MAPolGradG'}), 
	\begin{align}
	&\nabla_{\theta_i}J_i(\bm{\theta}) - \dot{\theta}_i =\nabla_{\theta_i}J_i(\bm{\theta}) \nonumber\\
	&=\mathbb{E}_{\bm{\pi_\theta}}[\nabla_{\theta_i} \log\pi_{\theta_i} (a^i_t|o^i_t). \nonumber \\&\qquad \qquad(\gamma (V_i^{\pi_{\bm{\theta}}}(s_{t+1}) - V_{\omega_i(\bm{\theta})} (s_{t+1}))\nonumber\\&\qquad \qquad- (V_i^{\pi_{\bm{\theta}}}(s_{t}) - V_{\omega_i(\bm{\theta})}(s_t)))] \nonumber
	\end{align} 
	
	Using the Cauchy-Schwarz inequality, 
	\begin{align}
	||\nabla_{\theta_i}J_i(\bm{\theta}^*)||_2 &\leq |\gamma +1| .||V_i^{\pi_{\bm{\theta}}}(s) - V_{\omega_i(\bm{\theta})}(s)||_{\pi_{\bm{\theta}}}.\nonumber\\&\qquad \qquad \qquad ||\nabla_{\theta_i} \log\pi_{\theta_i}||_{\pi_{\bm{\theta}}}\nonumber\\
	&\leq C_i(\bm{\theta}^*)\mathcal{E}_i(\bm{\theta}^*)\label{CauchSchwIneq}
	\end{align}

	Each term on the right side of Eqn. (\ref{CauchSchwIneq}) is bounded. Thus, $J_i(\bm{\theta})$ converges for each agent $i$ in the original game $\mathcal{G}$, even though policies are synthesized in the modified game $\mathcal{G}'$. 
\end{proof}

Proposition \ref{BoundInOrigGame} demonstrates that the additional reward $F^i$ provided by SAM to guide the agents does not distract them from accomplishing the task objective that is originally specified by the environment reward $R^i$.

Given similar assumptions, we can obtain Corollary \ref{convergence} for single-agent cases and show guarantees on the convergence of Algorithm \ref{Algo1} using the theory of `two time-scale stochastic analysis'.
Corollary \ref{convergence} gives a bound on the error introduced as a result of approximating the value function $V_{M'}$ with $V_{M'}^{\omega}$ in the MDP $M'$.
This error term is small if the linear function family for $V_{M'}^{\omega}$ is rich. 
In fact, if the critic is updated in batches, a tighter bound can be achieved, as shown in Proposition 1 of \cite{Yang2018finite}. 

\begin{cor}\label{convergence} 
	Let $\mathcal{E}(\theta):=\norm{V^{\omega(\theta)}_{M'}(s)-V^{\pi_{\theta}}_{M'}(s)}_{\rho^{\pi_{\theta}}}$.
	Then, for any limit point $(\theta^*,\omega^*) :=\lim\limits_{T_{max} \to \infty}(\theta_{T_{max}},\omega_{T_{max}})\}$ of Algorithm \ref{Algo1}, 
	$\norm{\nabla_{\theta}J_M(\theta^*)}_2\leq C\mathcal{E}(\theta^*)$.
\end{cor}


Algorithm \ref{algo:MAAC+PBA} desecribes SAM. 
The shaping advice is specified as a difference of potential functions (\emph{Line 15}), and is added to the reward received from the environment. 
We use an advantage-based actor-critic, and use the TD-error to estimate this advantage (\emph{Line 16}). This is used to update the actor and critic parameters for each agent (\emph{Lines 18-19}).
\begin{algorithm}[!h]
	\small
	\caption{SAM: Shaping Advice in Multi-agent RL}
	\label{algo:MAAC+PBA}
	\begin{algorithmic}[1]

		\REQUIRE{For each agent $i$: differentiable policy $\pi_{\theta_i}$; differentiable value function $V_{\omega_i}$; shaping advice $\phi_i(s,a^i,a^{-i})$. Maximum episode length $T_{max}$.}
		
\textbf{Initialization}: \\		
	policy and value function parameters $\theta_i, \omega_i$ for all agents $i$, learning rates $\alpha^\theta, \alpha^\omega$.
		\STATE{$T = 0$}
		\REPEAT
		\STATE{$t \leftarrow -1$; $\phi_i(s_{-1},a^i_{-1},a^{-i}_{-1}) = 0$ for all $i$}
		\STATE{Initialize information $s_0 = [o_0^1,\dots,o_0^n]$}
		\REPEAT
		\STATE{$t \leftarrow t+1$}
		\FOR{agent $i=1$ to $n$}
		\STATE{sample $a^i_t \sim \pi_{\theta_i}(\cdot|o^i_t)$}
		\ENDFOR
		\STATE{Take action $a_t = [a^1_t,\dots,a^n_t]$, observe new information $s_{t+1}$ and obtain reward $r^i_t$ for each agent. Use $a_t$ to determine $\phi_i(s_t,a_t)$ for all agents}
		\IF{$s_{t+1}$ is terminal}
		\STATE{$V_{\omega_i} (s_{t+1}) =0$}
		\ENDIF
		\FOR{agent $i=1$ to $n$}
		\STATE{compute $F^i_t$ based on equations (\ref{LAPBA}) and (\ref{LBPBA})}
		\STATE{TD-error: $\delta^i_t:=r^i_t+F^i_t+\gamma V_{\omega_i} (s_{t+1}) - V_{\omega_i} (s_t)$}
		\STATE{compute $\tilde{\delta}^i_t$ based on equations (\ref{delta})}
		\STATE{Update actor: $\theta_i \leftarrow \Gamma_i[\theta_i + \alpha_t^\theta \tilde{\delta}^i_t  \nabla_{\theta_i} \log~\pi_{\theta_i}(a_t^i|o^i_t)]$}
		\STATE{Update critic: $\omega_i \leftarrow \omega_i - \alpha_t^\omega \delta^i_t \nabla_{\omega_i} V_{\omega_i} (s_t)$}
		\ENDFOR
		\UNTIL{$s_{t+1}$ is terminal}
		\STATE{$T \leftarrow T+1$}
		\UNTIL{$T>T_{max}$}
	\end{algorithmic}
\end{algorithm}

\begin{rem}
We note that our objective is to maximize the rewards that can be obtained by agents equipped with shaping advice. 
Our algorithms are termed to converge when the values of these rewards reaches a `steady-state'. 
This is distinct from game-theoretic notions of equilibrium (e.g., Nash equilibrium), as studied in \cite{jiang2020enhanced, yang2018mean}. 
The flavor of theoretical analysis of convergence that we adopt also allows better interpretability of our experimental evaluations, reported in Section VI. 
\end{rem}

%% file: Experiment.tex
\section{Experiments}\label{Experiments}
Our experiments study the incorporation of shaping advice in both single-agent and multi-agent environments. 
The code for our experimental evaluations is available at \textcolor{blue}{\url{https://github.com/baicenxiao/Shaping-Advice}}. 

\subsection{Shaping advice for single-agent RL}
In the single-agent case, we seek to compare the performance of actor-critic without shaping advice, SAS-Uniform and SAS-NonUniform. 
For both SAS-Uniform and SAS-NonUniform, actor-critic \cite{sutton2018reinforcement} is adopted as the base RL algorithm.
We consider two setups. 
The first is a \emph{Puddle-Jump Gridworld} \cite{marom2018belief}, where state and action spaces are discrete. 
The second environment is a continuous state and action space \emph{mountain car} \cite{brockman2016openai}. 
In each experiment, we compare the rewards received by the agent when it uses the following schemes: \emph{i)}: actor-critic with sparse rewards (Sparse); \emph{ii)}: SAS-Uniform; \emph{iii)}: SAS-NonUniform.
For SAS-NonUniform, we apply it in a look-back way since it does not require the estimation of future actions.

\subsubsection{Puddle-Jump Gridworld}
\begin{figure}
	\centering
	\includegraphics[width=1.5 in]{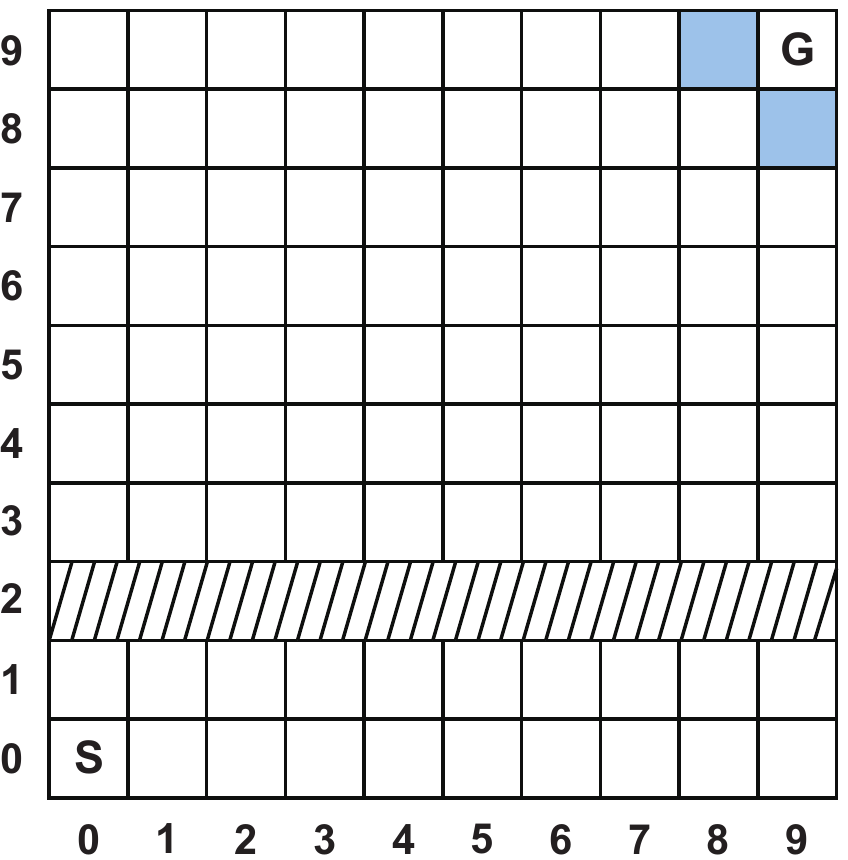}
	\caption{Schematic of the puddle-jump gridworld. The state of the agent is its position $(x,y)$. The shaded row (row $2$) represents the puddle the agent should jump over. The two blue grids denote states that are indistinguishable to the agent. The agent can choose an action from the set $\{up, down, left, right, jump\}$ at each step.}\label{GridWorld}
\end{figure}

Figure \ref{GridWorld} depicts the \emph{Puddle-jump gridworld} environment as a 10x10 grid. 
The state space is $s=(x,y)$ denoting the position of the agent in the grid, where $x,y \in \{0,1,\dots,9\}$. 
The goal of the agent is to navigate from the start state $S= (0,0)$ to the goal $G=(9,9)$. 
At each step, the agent can choose from actions in the set $A = \{up, down, left, right, jump\}$. 
There is a \emph{puddle} along row $2$ which the agent should jump over. Further, the states $(9,8)$ and $(8,9)$ (blue squares in Figure \ref{GridWorld}) are indistinguishable to the agent. 
As a result, any optimal policy for the agent will be a stochastic policy.

If the $jump$ action is chosen in rows $3$ or $1$, the agent will land on the other side of the puddle with probability $p_j$, and remain in the same state otherwise. 
This action chosen in other rows will keep the agent in its current state. 
Any action that will move the agent off the grid will keep its state unchanged. 
The agent receives a reward of $-0.05$ for each action, and $+1000$ for reaching $G$.

When using SAS-Uniform, we set $\phi^{U}(s):=u_0$ for states in rows $0$ and $1$,  and $\phi^{U}(s):=u_1$ for all other states. 
We need $u_1>u_0$ to encourage the agent to jump over the puddle. 
Unlike SAS-Uniform, SAS-NonUniform can provide the agent with more information about the actions it can take. 
We set $\phi^{NU}(s,a)$ to a `large' value if action $a$ at state $s$ results in the agent moving closer to the goal according to the $\ell_1$ norm, $\big(|G-x|+|G-y|
\big)$. 
We additionally stipulate that $\frac{1}{|A|}\sum_{a\in A}\phi^{NU}(s,a) = \phi^{U}(s)$. That is, the state potential of SAS-NonUniform is the same as that of SAS-Uniform under a uniform distribution over actions. 
This is to ensure a fair comparison between SAS-Uniform and SAS-NonUniform.

In our experiment, we set the discount factor $\gamma=1$. 
Since the dimensions of the state and action spaces is not large, we do not use a function approximator for the policy $\pi$. 
A parameter $\theta_{s,a}$ is associated to each state-action pair, and the policy is computed as:
$\pi_{\theta}(a|s)=\frac{\exp(\theta_{s,a})}{\sum_{a\in A}\exp(\theta_{s,a})}$.
We fix $\alpha^{\omega}= 0.001$, and $\alpha^{\theta}= 0.2$ for all cases. 
From Figure \ref{cliff_results}, we observe that SAS-NonUniform scheme performs the best, in that the agent converges to the goal in \textbf{five times} fewer episodes ($25$ vs. $125$ episodes) than A2C without advice (Sparse). 
When A2C is augmented with SAS-Uniform, convergence to the goal is slightly faster than without any reward shaping. 
\begin{figure}
	\centering
	\includegraphics[width=2.8 in]{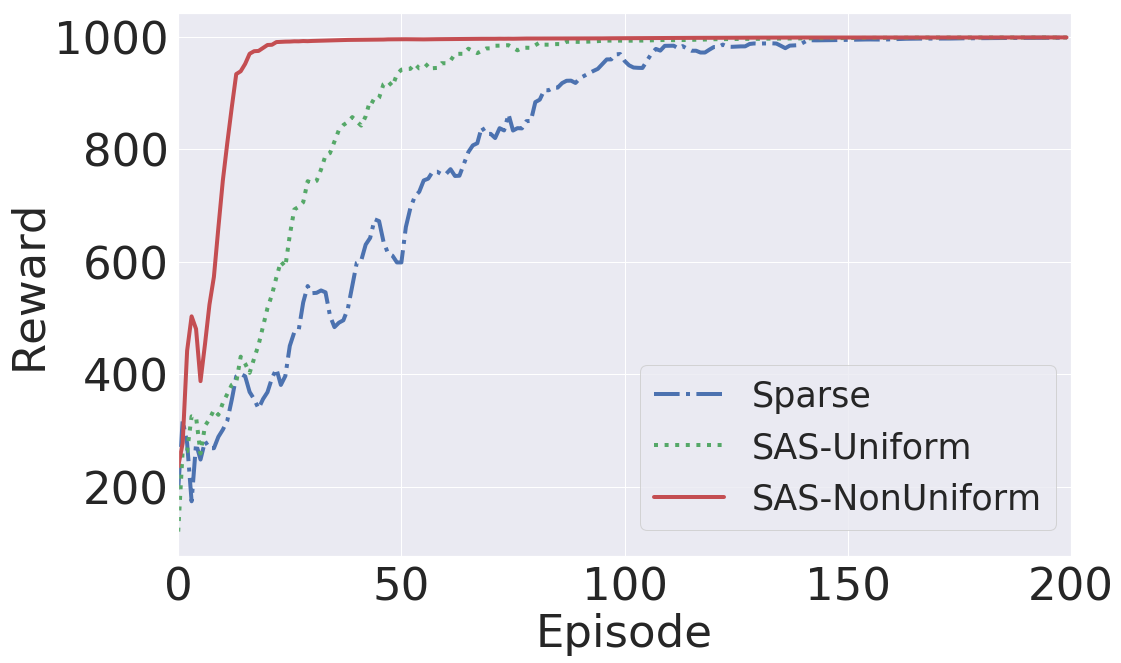}
	\caption{Average rewards in puddle-jump gridworld when jump success probability $p_j=0.2$. The baseline is the advantage actor-critic without advice.}\label{cliff_results}
\end{figure}

A smaller jump success probability $p_j$ is an indication that it is more difficult for the agent to reach the goal state $G$. 
Figure \ref{cliff_results2} shows that SAS-NonUniform results in the highest reward for a more difficult task (lower $p_j$), when compared with the other reward schemes.
\begin{figure}
	\centering
	\includegraphics[width=2.8 in]{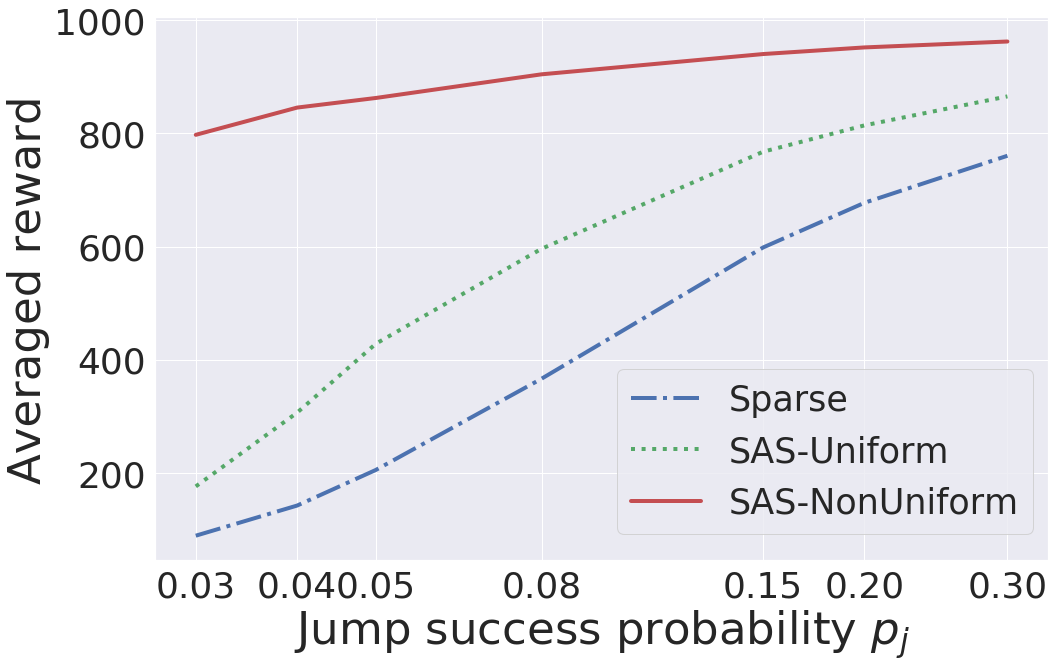}
	\caption{Average reward for the first 100 episodes with respect to the jump success probability $p_j$.}\label{cliff_results2}
\end{figure}
%
\subsubsection{Continuous Mountain Car}

In the mountain car (MC) environment, an under-powered car in a valley has to drive up a steep hill to reach the goal. In order to achieve this, the car should learn to accumulate momentum. 
A schematic for this environment is shown in Figure \ref{MountainCarFig}.
\begin{figure}
	\centering
	\includegraphics[width=3 in]{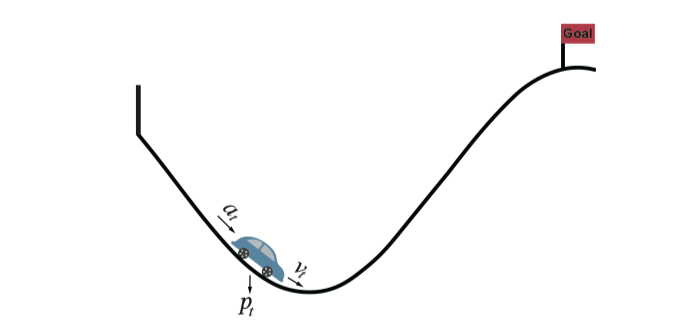}
	\caption{Schematic of the mountain-car environment. The agent's state is represented by its position $p_t$ (along the $x-$coordinate) and velocity $v_t$. The action $a_t$ is a force applied to the car. The goal is marked as a flag.}\label{MountainCarFig}
\end{figure}
This environment has continuous state and action spaces. 
The state $s=(p,v)$ denotes position $p \in [-1.2,0.6]$ and velocity $v \in [-0.07,0.07]$. 
The action $a \in [-1,+1]$. 
The continuous action space makes it difficult to use classic value-based methods, such as Q-learning and SARSA-learning. 
The reward provided by the environment depends on the action and whether the car reaches the goal. 
Specifically, once the car reaches the goal it receives $+100$, and before that, the reward at time $t$ is $-|a_t|^2$. 
This reward structure therefore discourages the waste of energy. 
This acts as a barrier for learning, because there appears to be a sub-optimal solution where the agent remains at the bottom of the valley.
Moreover, the reward for reaching the goal is significantly delayed, which makes it difficult for the conventional actor-critic algorithm to learn a good policy. 

One choice of a potential function while using SAS-Uniform in this environment is $\phi^{U}(s_t):=p_t+P$. 
Since the position $p_t$ can take values in $[-1.2,0.6]$, the offset $P$ is chosen so that the potential $\phi^{U}(s_t)$ will always be positive. 
We use $P = 2$. 
An interpretation of this scheme is: \emph{`state value is larger when the car is horizontally closer to the goal.'} 
The SAS-NonUniform scheme we use for this environment encourages the accumulation of momentum by the car-- the direction of the action is encouraged to be the same as the current direction of the car's velocity. 
In the meanwhile, we discourage inaction. 
Mathematically, the potential advice function has a larger value if \emph{ $a_t\neq 0$}. 
We let $\phi^{NU}(s_t,a_t)=1$, if $a_tv_t >0$, and $\phi^{NU}(s_t,a_t)=0$, otherwise.

In our experiments, we set $\gamma= 0.99$. 
To deal with the continuous state space, we use a neural network (NN) as a  function approximator. 
The policy distribution $\pi_{\theta}(a|s)$ is approximated by a normal distribution, the mean and variance of which are the outputs of the NN. 
The value function is also represented by an NN. 
We set $\alpha^{\theta}=1\times 10^{-5}$ and $\alpha^{\omega}=5.6\times 10^{-4}$, and use Adam \cite{Adam} to update the NN parameters. 
The results we report are averaged over 10 different environment seeds.
\begin{figure}
	\centering
	\includegraphics[width=2.6 in]{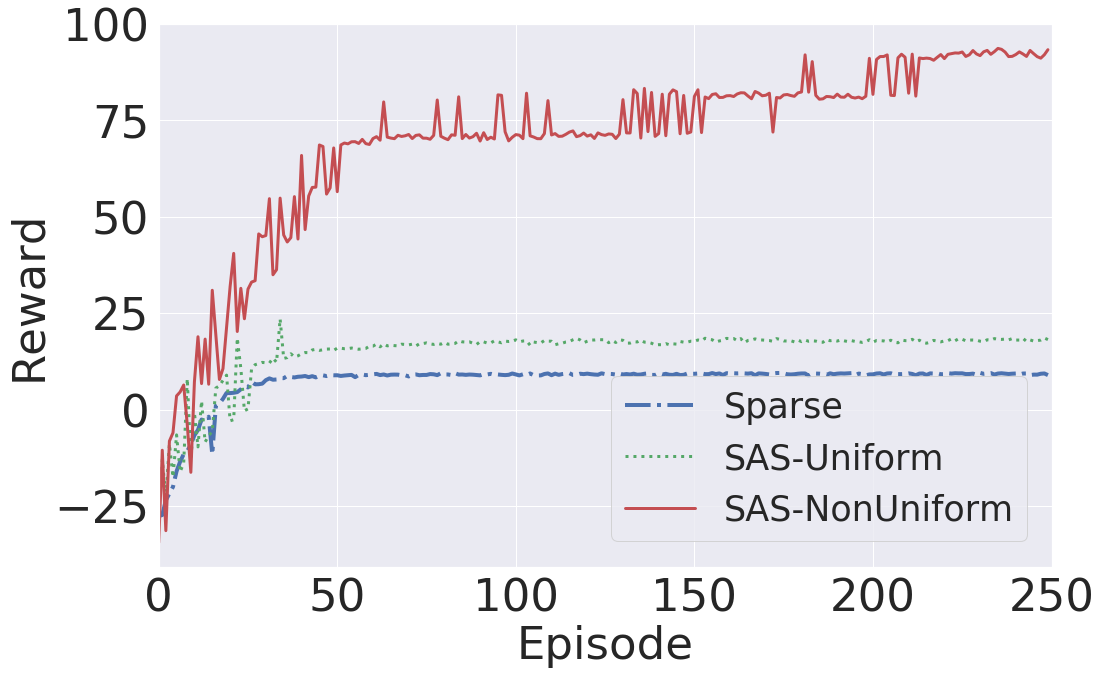}
	\caption{Average rewards for continuous mountain car problem (averaged over 10 different environment random seeds). The baseline is the A2C without advice.}\label{MountainCar_results}
\end{figure}
%

Our experiments indicate that the policy makes the agent converge to one of two points: the goal, or remain stationary at the bottom of the valley. 
We observe that when learning with the base algorithm (A2C with sparse rewards), the agent is able to reach the goal only in $10\%$ of the trials (out of 10 trials), and was stuck at the sub-optimal solution for the remaining trials. 
With SAS-Uniform, the agent could converge correctly in only $20\%$ of the trials. 
This is because the agent might have to take an action that moves it away from the goal in order to accumulate momentum. 
However, the potential function $\phi^{U}(\cdot)$ discourages such actions. 
SAS-NonUniform performs the best, where we observed that the agent was able to reach the goal in $100\%$ of the trials. 
Figure \ref{MountainCar_results} reflects these observations, where we see that SAS-NonUniform results in the agent obtaining the highest rewards. 
\begin{table*}[]
	\centering
	\begin{tabular}{|c|c|c|}
		\hline
		\textbf{Task}            & \textbf{$\phi_i(s_t,a^i_t,a^{-i}_t)$: SAM-Uniform} & \textbf{$\phi_i(s_t,a^i_t,a^{-i}_t)$: SAM-NonUniform} \\ \hline
		CN & $\alpha_{1}exp(-\beta_{1}\sum_{j=1}^N dist(s_t^j,L_j))$    & $ -M_{1}\theta_{{a^i_tL_i}}+\alpha_{2}exp(-\beta_{2}\sum_{j=1}^N dist(s_t^j,L_j))$ \\ \hline
		PD & $ \alpha_{3}exp(-\beta_{2}\sum_{j=1}^N dist(s_t^j,L_j))$    & $-M_{2}\theta_{{a^i_tL_i}}+\alpha_{4}exp(-\beta_{4}\sum_{j=1}^N dist(s_t^j,L_j))$\\ \hline
		PP & $ \alpha_{5}exp(-\beta_{5}\sum_{j=1}^N dist(s_t^{pred_j},s_t^{prey}))$   & $-M_{3}\sum_{j=1}^N \theta_{{a^{pred_j}_t s_t^{prey}}}+\alpha_{6}exp(-\beta_{6}\sum_{j=1}^N dist(s_t^{pred_j},s_t^{prey}))$ \\ \hline
	\end{tabular}\caption{Shaping advice, $F^i_t$ provided by SAM is given by Equation (\ref{LAPBA}) or (\ref{LBPBA}). The table lists the potential functions used in the Cooperative Navigation (CN), Physical Deception (PD), and Predator-Prey (PP) tasks. $L_j$ is the landmark to which agent $j$ is \emph{anchored} to. $dist(\cdot,\cdot)$ denotes the Euclidean distance. $\theta_{{a^j_tL_j}} \in [0, \pi]$ is the angle between the direction of the action taken by agent $j$ and the vector directed from its current position to $L_j$. In \emph{SAM-Uniform}, advice for every action of the agents for a particular $s_t$ is the same. 
		In \emph{SAM-NonUniform}, agents are additionally penalized if their actions are not in the direction of their target. In each case, $F^i_t$ is positive when agents take actions that move it towards their target.
	}\label{TableSAM}
\end{table*}

\subsection{Shaping Advice for Multi-agent RL}
This section describes the multi-agent tasks that we evaluate SAM on, and these include tasks with cooperative and competitive objectives. 
In each case, the rewards provided to the agents are sparse, which affects the agents' ability to obtain immediate feedback on the quality of their actions at each time-step. 
Shaping advice provided by SAM is used to guide the agents to obtain higher rewards than in the case without advice. 
We conclude the section by presenting the results of our experiments evaluating SAM on these tasks. 

\subsubsection{Task Descriptions and Shaping Advice}
\begin{figure}[!h]
	\centering
	\includegraphics[width=3.30 in]{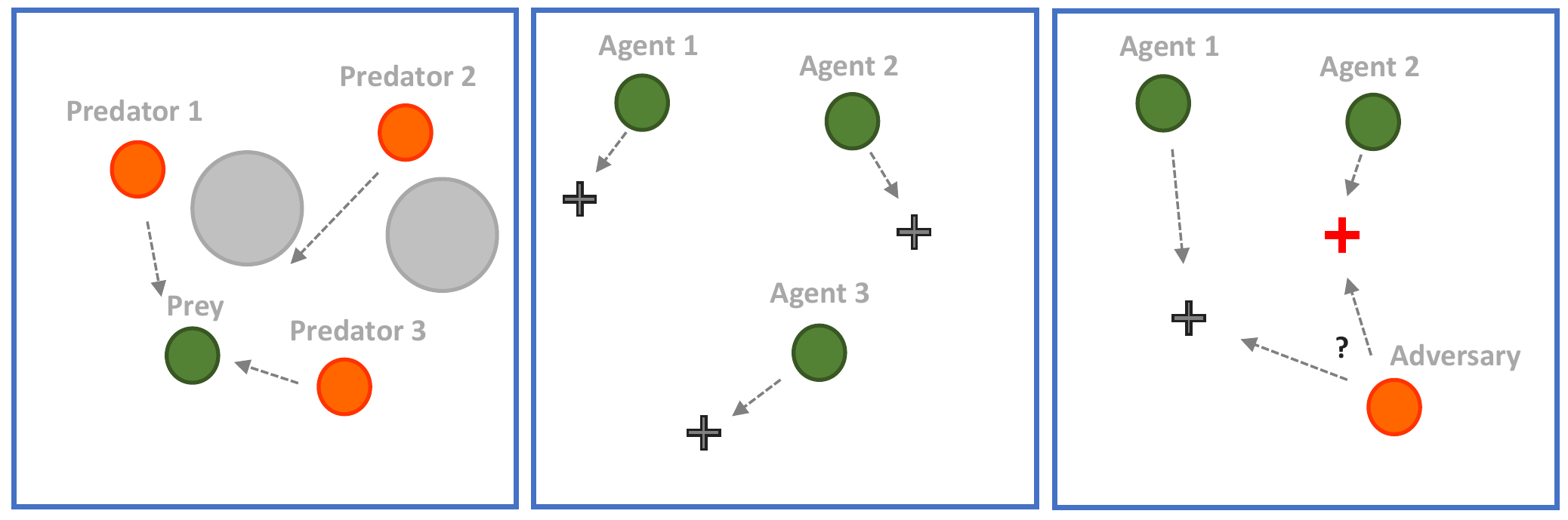} 
	\caption{Representations of tasks from the Particle World Environment \cite{lowe2017multi} that we study. (\emph{Left to Right}) Predator-Prey (PP), Cooperative Navigation (CN), and Physical Deception (PD). In PP, predators (red) seek to catch the prey (green) while avoiding obstacles (grey). 
	In CN, agents (green) each seek to navigate to a different landmark ($\times$) and are penalized for collisions with each other. In PD, one of the agents (green) must reach the true landmark (red $\times$), while preventing the adversary from reaching this landmark. In all tasks, rewards are \emph{sparse}. Agents receive a reward or penalty only when a corresponding reachability or collision criterion is satisfied. 
	}\label{PartWorldEnv}
\end{figure}

We examine three tasks from the \emph{Particle World} environment \cite{lowe2017multi} where multiple agents share a two-dimensional space with continuous states and actions. 
An illustration of the tasks is shown in Figure \ref{PartWorldEnv}, and we describe them below. 

\emph{Predator-Prey}:  
This task has $N$ predator agents who cooperate to capture $1$ faster-moving prey. 
Predators are rewarded when one of them collides with the prey, while the prey is penalized for the collision. 
The reward at other times is zero. 
Two landmarks impede movement of the agents. 
%
%

\emph{Cooperative Navigation}:
This task has $N$ agents and $N$ landmarks. 
Agents are each rewarded $r$ when an agent reaches a landmark, and penalized for collisions with each other. 
The reward at other times is zero. 
Therefore, the maximum rewards agents can obtain is $rN$.  
Thus, agents must learn to \emph{cover} the landmarks, and not collide with each other. 

\emph{Physical Deception}: 
This task has $1$ adversary, $N$ agents, and $N$ landmarks. 
Only one landmark is the true target. 
Agents are rewarded when any one reaches the target, and penalized if the adversary reaches the target. 
At all other times, the agents get a reward of zero. 
An adversary also wants to reach the target, but it does not know which landmark is the target landmark. 
Thus, agents have to learn to split up and cover the landmarks to deceive the adversary. 

In each environment, SAM provides shaping advice to guide agents to obtain a higher positive reward. 
This advice is augmented to the reward received from the environment. 
%
The advice is a heuristic given by a difference of potential functions (Equations (\ref{LAPBA}) or (\ref{LBPBA})), and only needs to be specified \emph{once} at the start of the training process. 
In the \emph{Cooperative Navigation} and \emph{Physical Deception} tasks, we \emph{anchor} each agent to a (distinct) landmark. 
The shaping advice will then depend on the distance of an agent to the landmark it is anchored to. 
Although distances computed in this manner will depend on the order in which the agents and landmarks are chosen, we observe that it \emph{empirically} works across multiple training episodes where positions of landmarks and initial positions of agents are generated randomly. 
The advice provided by SAM is positive when agents move closer to landmarks they are anchored to. 
In the absence of anchoring, they may get distracted and move towards different landmarks at different time steps.  
Anchoring results in agents learning to cover landmarks faster.

We consider two variants of advice for each task. 
In \textbf{\emph{SAM-Uniform}}, the advice for every action taken is the same. 
In \textbf{\emph{SAM-NonUniform}}, a higher weight is given to some `good' actions over others for each $s_t$. 
We enumerate the advice for each task in Table \ref{TableSAM}. 
We use MADDPG as the base RL algorithm \cite{lowe2017multi}. 
We compare the performance of agents trained with SAM (SAM-Uniform or SAM-NonUniform) to the performance of agents trained using the sparse reward from the environment. 
We also compare SAM with a state-of-the-art reward redistribution technique introduced in \cite{gangwani2020learning}  called Iterative Relative Credit Assignment (IRCR). 

\subsubsection{Implementation details}
When we tested SAM-NonUniform, we used look-back advice following Equation (\ref{LBPBA}).
This was done to avoid estimating a `future' action, i.e. $a^i_{t+1}$, at each time-step (since the replay buffer contains tuples of the form $(s_t, a^1_t,\dots,a^n_t, r^1_t,\dots,r^n_t, s_{t+1})$). 
Noisy estimates of $a^i_{t+1}$ can cause oscillations in the rewards obtained by the agent. 

We adopt the same hyperparameter values and network architectures as used in \cite{lowe2017multi}. 
We let $\Gamma_i$ \emph{Line 18} of the SAM Algorithm be the identity matrix. 
We list values of $\alpha_{\cdot}, \beta_{\cdot}, M_{\cdot}$ (from Table 1 in the main paper) that were used in our experiments.
%
\textbf{Sparse Reward Setting}: 
For results reported in Figures \ref{BarPlotScore} and \ref{FigGraphs}, we use the following parameters for the shaping advice:

\begin{center}
	\begin{tabular}{l l l l}
		
		CN ($N=3$) & $\alpha_1 = \alpha_2 = 100$ & $\beta_1 = \beta_2 =1$ & $M_1 = 1$ \\
		CN ($N=6$) & $\alpha_1 = \alpha_2 = 1000$ & $\beta_1 = \beta_2 =1$ & $M_1 = 10$ \\
		PD ($N=2$) & $\alpha_3 = \alpha_4 = 500$ & $\beta_3 = \beta_4 = 1$ & $M_2 = 1$\\
		PD ($N=4$) & $\alpha_3 = \alpha_4 = 500$ & $\beta_3 = \beta_4 = 1$ & $M_2 = 10$\\
		PP ($N=3$)& $\alpha_5 = \alpha_6 = 100$ & $\beta_5 = \beta_6 = 1$ & $M_3 = 1$\\ \\
	\end{tabular}
\end{center}

%
%

\textbf{Other forms of shaping rewards}: 
In the three tasks that we evaluated, we observed that a `linear' distance-based advice (i.e., advice of the form $\phi_i(s_t,a^i_t,a^{-i}_t) :=\sum_j dist(s_t^j,L_j)$) did not work. 
From Equations (\ref{LAPBA}) and (\ref{LBPBA}), using this form of advice, an agent would get the same additional reward when it took a step towards the target, independent of its distance to the target. 
For example, an agent $100$ steps from the target would get the same shaping advice if it took one step towards the target as an agent who takes a step towards the target from a state that is $50$ steps from the target. \\

\subsubsection{Results}

\begin{figure}[!h]
	\centering
	\includegraphics[width=3.20 in]{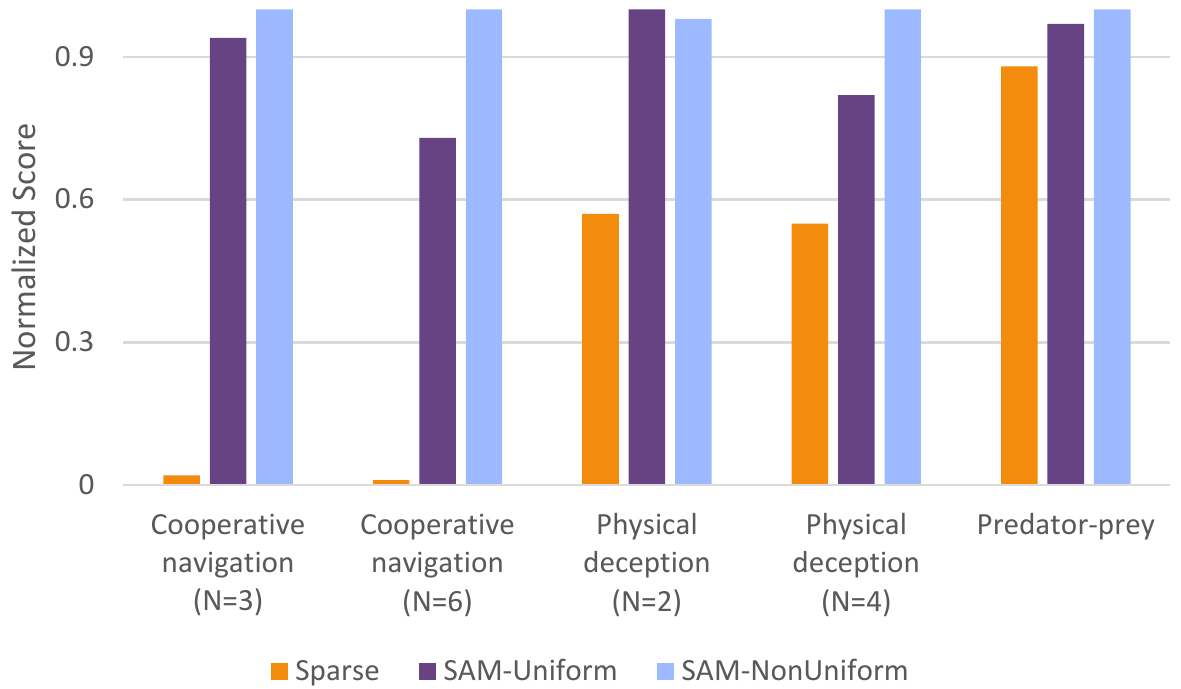} 
	\caption{Comparison between SAM (SAM-NonUniform (blue) or SAM-Uniform (purple)) augmented to MADDPG and classical MADDPG policies (orange) on cooperative and competitive tasks with sparse rewards. The \textbf{score} for a task is the average agent reward in cooperative tasks, and the \emph{average agent advantage} ($=$ agent reward $-$ adversary reward) in competitive tasks. Each bar cluster shows Normalized $0-1$ scores, averaged over the last $1000$ training episodes. Higher score is better. SAM-NonUniform outperforms SAM-Uniform and the classical MADDPG baseline by a larger margin when there are more agents in the cooperative navigation and physical deception tasks.
	The scores of IRCR are not shown here, since it performs consistently worse than other approaches in these three tasks.
	}\label{BarPlotScore}
\end{figure}

\begin{figure*}
	\begin{subfigure}{0.31\textwidth}
		\includegraphics[width=\linewidth]{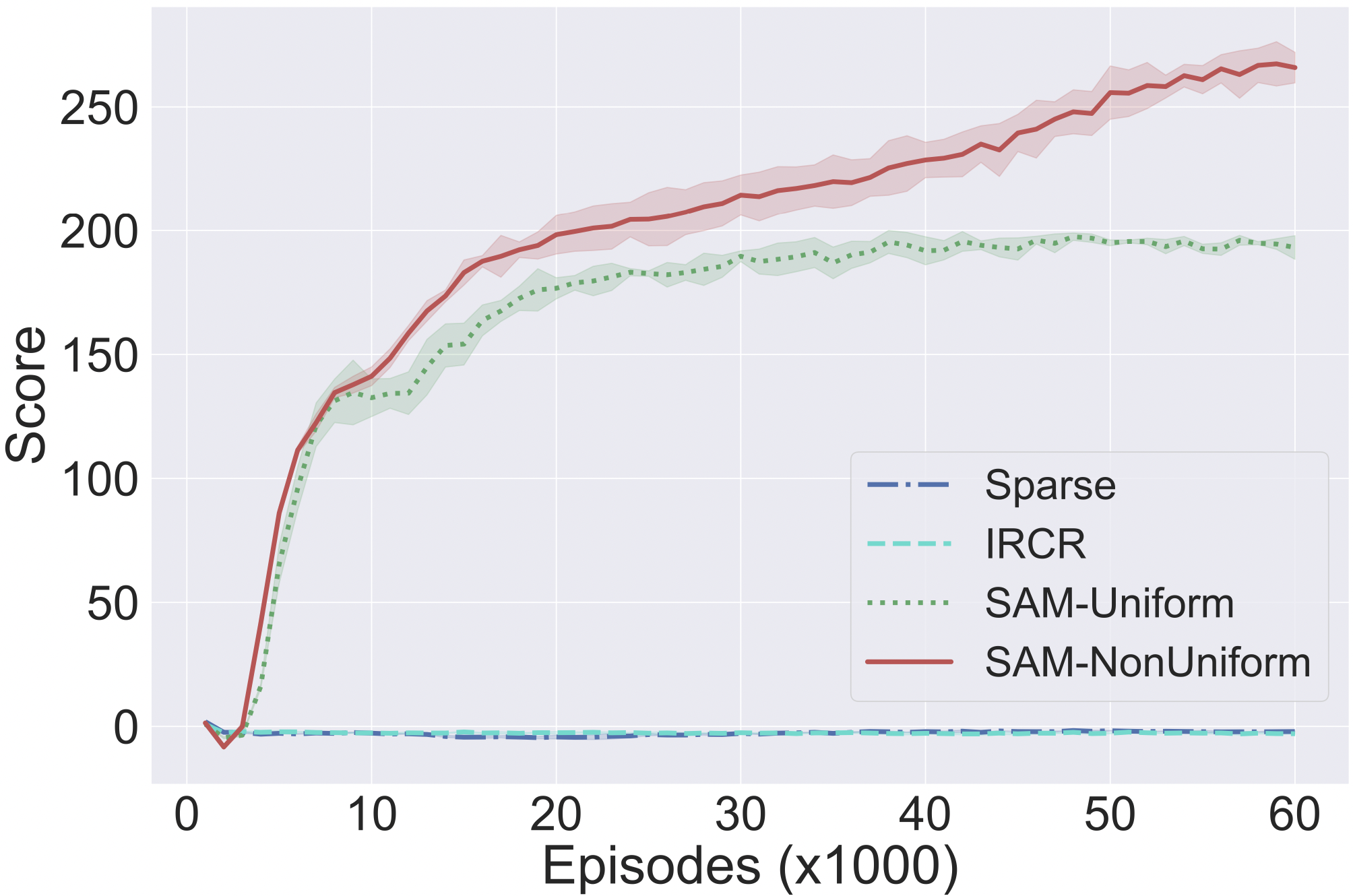}
		\caption{Cooperative Navigation ($N=6$)} \label{fig:1a}
	\end{subfigure}%
	\hspace*{\fill}   
	\begin{subfigure}{0.31\textwidth}
		\includegraphics[width=\linewidth]{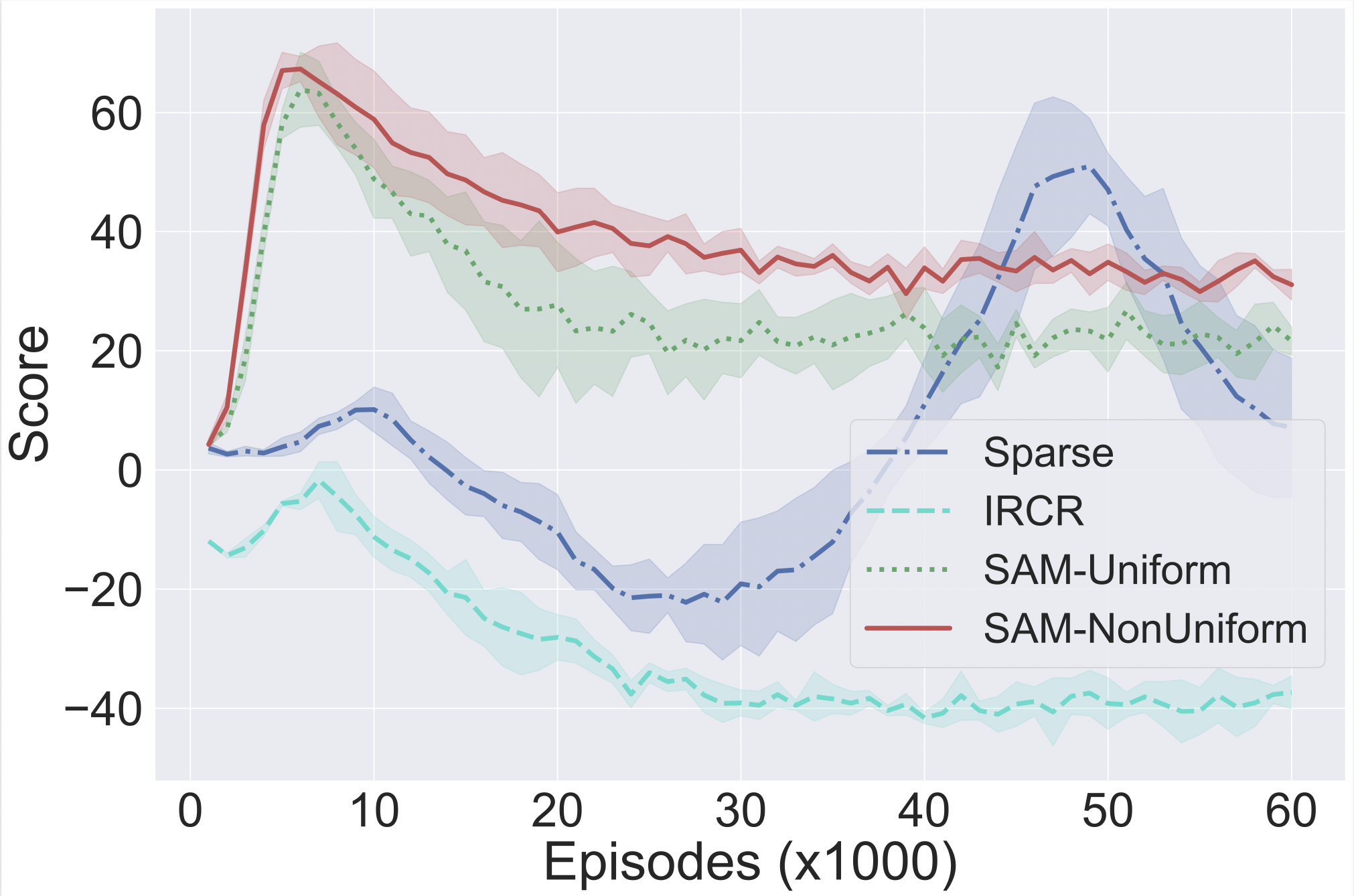}
		\caption{Physical Deception ($N=4$)} \label{fig:1b}
	\end{subfigure}%
	\hspace*{\fill}   
	\begin{subfigure}{0.31\textwidth}
		\includegraphics[width=\linewidth]{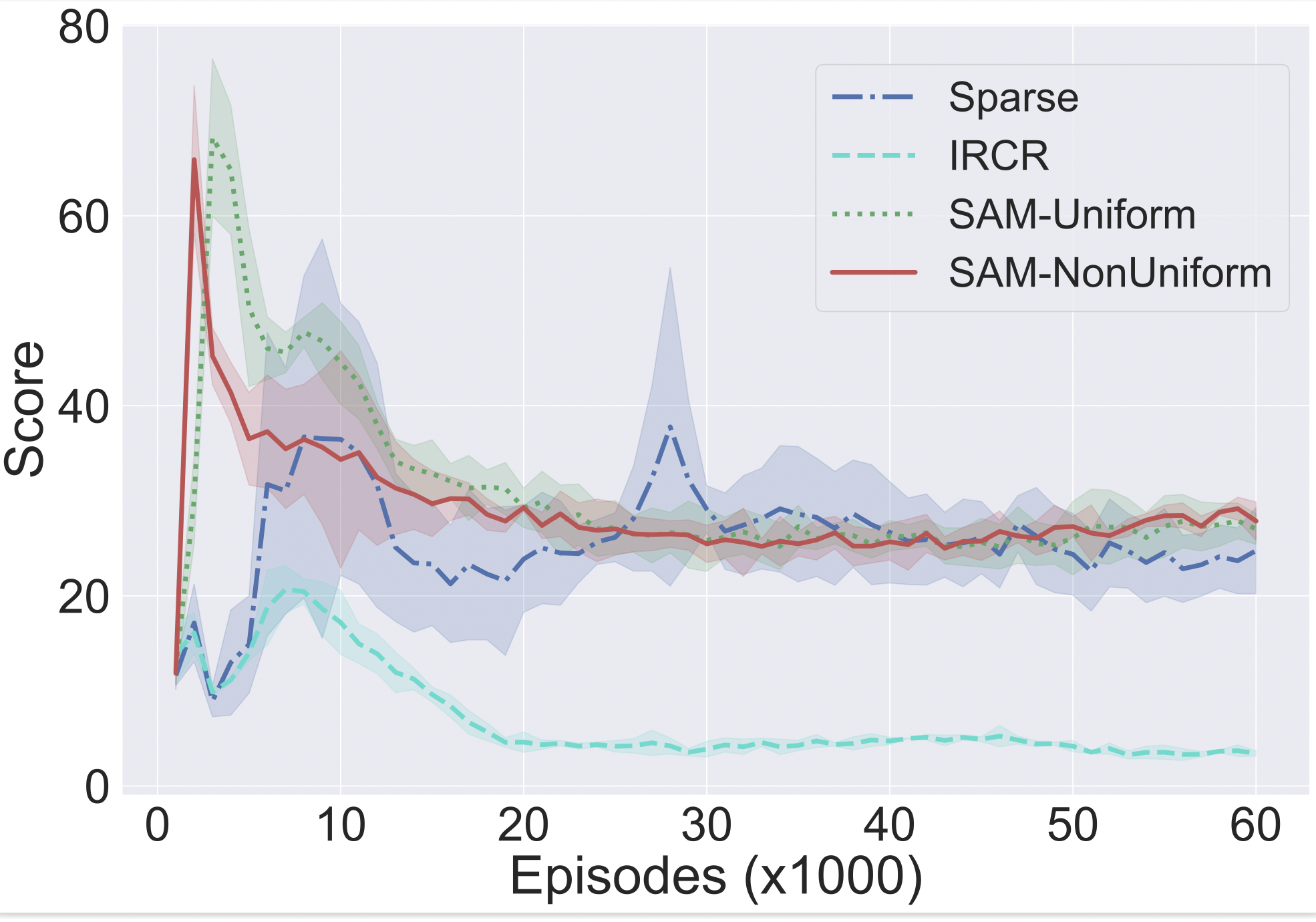}
		\caption{Predator-Prey ($3$ pred., $1$ prey.)} \label{fig:1c}
	\end{subfigure}
	
	\caption{Average and variance of scores when agents use SAM-NonUniform (blue), SAM-Uniform (purple), IRCR (green) and sparse rewards (orange). SAM-NonUniform results in the highest average scores. SAM-Uniform compares favorably, and both significantly outperform agents trained using only sparse rewards. 
		IRCR is not able to guide agents to obtain higher rewards in all three tasks.
	} \label{FigGraphs}
\end{figure*}
%

Figure \ref{BarPlotScore} shows $0-1$ normalized scores, averaged over the last 1000 training episodes, comparing SAM augmented to MADDPG and classical MADDPG policies. 
The \textbf{score} for a task is the average agent reward in cooperative tasks, and the \emph{average agent advantage} ($=$ agent $-$ adversary reward) 
in competitive tasks \cite{wen2019probabilistic}. 
Agents equipped with \emph{SAM-NonUniform} have the best performance. 
This is because SAM-NonUniform provides specific feedback on the quality of agents' actions. 
\emph{SAM-Uniform} also performs well in these tasks. 
SAM-NonUniform outperforms SAM-Uniform and the classical MADDPG baseline by a significant margin when there are more agents. 

In cooperative navigation, when rewards are sparse, the agents are not able to learn policies that will allow them to even partially cover the landmarks using the MADDPG baseline. 
In comparison, SAM guides agents to learn to adapt to each others' policies, and cover all the landmarks. 
SAM-NonUniform results in much higher rewards than other methods in the complex task with $N=6$ agents. 

We observe a similar phenomenon in physical deception, where SAM guides agents to learn policies to cover the landmarks. 
This behavior of the agents is useful in deceiving the adversary from moving towards the true landmark, thereby resulting in lower rewards for the adversary. 
Therefore, agent advantage is higher with SAM. 

In the predator-prey task, we see that the performance of SAM is comparable to MADDPG. 
We believe that this is because this task might have a well-defined and unique \emph{equilibrium} to which agent policies eventually converge. 

Figure \ref{FigGraphs} shows the average and variance of the \textbf{score} during different stages of the training process.
The score for a task is the average agent reward in cooperative tasks, and the \emph{average agent advantage} ($=$ agent $-$ adversary reward) in competitive tasks \cite{wen2019probabilistic}. 
In terms of agent scores averaged over the last $1000$ training episodes,
agents equipped with \emph{SAM-NonUniform} have the best performance. 
This is because SAM-NonUniform provides specific feedback on quality of agents' actions. 
\emph{SAM-Uniform} also performs well in these tasks. 

In cooperative navigation, when agents use only the sparse rewards from the environment, the agents are not able to learn policies that will allow them to even partially cover the landmarks. 
In comparison, SAM guides agents to learn to adapt to each others' policies, and cover all the landmarks. 
A similar phenomenon is observed in physical deception, where SAM guides agents to learn policies to cover the landmarks. 
This behavior of the agents is useful in deceiving the adversary from moving towards the true landmark, thereby resulting in lower final rewards for the adversary. 


We additionally compare the performance of SAM with a technique called IRCR that was introduced in \cite{gangwani2020learning}. 
We observe that agents using IRCR receive the lowest scores in all three tasks. 
We believe that a possible reason for this is that in each training episode, IRCR accumulates rewards till the end of the episode, and then uniformly redistributes the accumulated reward along the length of the episode. 
A consequence of this is that an agent may find it more difficult to identify the time-step when it reaches a landmark or when a collision occurs. 
For example, in the \emph{Predator-Prey} task, suppose that the length of an episode is $T_{ep}$. 
Consider a scenario where one of the predators collides with the prey at a time $T < T_{ep}$, and subsequently moves away from the prey. 
When IRCR is applied to this scenario, the redistributed reward at time $T$ will be the same as that at other time steps before $T_{ep}$. 
This property makes it difficult to identify critical time-steps when collisions between agents happen.  

The authors of \cite{lowe2017multi} observed that agent policies being unable to adapt to each other in competitive environments resulted in oscillations in rewards. 
Figure \ref{fig:1b} indicates that SAM is able to alleviate this problem. 
Policies learned by agents using SAM in the physical deception task 
result in much smaller oscillations in the rewards than when using sparse rewards alone. 

%

%% file: Conclusion.tex
\section{Conclusion}\label{Conclusion}

This paper presented a comprehensive framework to incorporate domain knowledge through shaping advice in single- and multi-agent reinforcement learning environments with sparse rewards. 
The shaping advice for each agent was a heuristic specified as a difference of potential functions, and was augmented to the reward provided by the environment. 
The modified reward signal provided agents with immediate feedback on the quality of the actions taken at each time-step. 
In the single-agent case, our algorithm SAS enabled the agent to obtain higher rewards, and was able to reach a target state more successfully than without a shaping reward. 
For the multi-agent setting, SAM used the centralized training with decentralized execution paradigm to efficiently learn decentralized policies for each agent that used only their individual local observations. 
We showed through theoretical analyses and experimental validation that shaping advice provided by SAS and SAM did not distract agents from accomplishing task objectives specified by the environment reward. 
Using SAS or SAM allowed the agents to obtain a higher average reward in fewer number of training episodes. 
 
Future research will aim to extend the scope of shaping advice based techniques. 
Two research questions are: 
i) can we design principled methods based on SAS and SAM in order to adaptively learn shaping advice (rather than being fixed for the duration of training)?, ii) how might such an adaptive procedure affect the sample efficiency and number of training episodes required? 
We believe that answers to these questions will be a step towards broadening the application of shaping advice to more challenging real-world RL environments. 
